\documentclass[a4paper,11pt]{scrartcl}

\usepackage[format=plain, indention=0.2cm]{caption}

\usepackage{amsmath,amssymb,amsthm}
\usepackage{hyperref, cleveref}
\usepackage{thm-restate}

\usepackage{lmodern,anyfontsize}

\usepackage{graphicx}

\usepackage{booktabs}

\usepackage[]{adjustbox, diagbox}

\usepackage[utf8]{inputenc}
\usepackage[noend,linesnumbered,ruled,lined]{algorithm2e}

\usepackage{mathrsfs}
 
\usepackage{enumitem}

\mathchardef\mhyphen="2D

\usepackage{xcolor}
\newtheorem{theorem}{Theorem}
\newtheorem{lemma}[theorem]{Lemma}
\newtheorem{definition}[theorem]{Definition}

\newtheorem{corollary}[theorem]{Corollary}
\newtheorem{problem}[theorem]{Problem}

\usepackage{todonotes}

\usepackage[]{multirow, array}

\newcommand{\Path}{\mathsf{Path}}
\newcommand{\ch}{\mathsf{ch}}
\newcommand{\height}{\mathsf{height}}

\SetArgSty{textnormal}

\usepackage{tikz}
\usetikzlibrary{calc}
\usetikzlibrary{decorations.pathmorphing}

\usepackage{comment}

\usetikzlibrary{external}
\tikzsetexternalprefix{external/}
\tikzset{external/only named=true}

\title{Connected $k$-Center and $k$-Diameter Clustering\thanks{This work has been funded by the Deutsche Forschungsgemeinschaft (DFG, German
Research Foundation) – 390685813; 459420781 and by the Lamarr Institute for Machine Learning and Artificial Intelligence
lamarr-institute.org.
}}

\author{ Lukas Drexler\thanks{Heinrich-Heine Universität Düsseldorf, Germany, e-Mail-addresses: \href{lukas.drexler@hhu.de}{lukas.drexler@hhu.de}, \href{mschmidt@hhu.de}{mschmidt@hhu.de}, \href{julian.wargalla@hhu.de}{julian.wargalla@hhu.de}} 
\and
Jan Eube\thanks{University of Bonn, Germany, e-Mail-addresses: \href{mailto:eube@informatik.uni-bonn.de}{eube@informatik.uni-bonn.de}, \href{mailto:s6dorein@uni-bonn.de}{s6dorein@uni-bonn.de}, \href{mailto:roeglin@cs.uni-bonn.de}{roeglin@cs.uni-bonn.de}} 
\and
Kelin Luo\thanks{University at Buffalo, USA, \href{kelinluo@buffalo.edu}{kelinluo@buffalo.edu}} 
\and 
Dorian Reineccius\footnotemark[3]
\and
Heiko R\"oglin\footnotemark[3]
\and
Melanie Schmidt\footnotemark[2] 
\and
Julian Wargalla\footnotemark[2]}
\date{}

\begin{document}

\pagenumbering{Alph}
\maketitle
\thispagestyle{empty}

\begin{abstract}
Motivated by an application from geodesy, we study the \emph{connected $k$-center problem} and the \emph{connected $k$-diameter problem}. These problems arise from the classical $k$-center and $k$-diameter problems by adding a side constraint. For the side constraint, we are given an undirected \emph{connectivity graph} $G$ on the input points, and a clustering is now only feasible if every cluster induces a connected subgraph in $G$. Usually in clustering problems one assumes that the clusters are pairwise disjoint. We study this case but additionally also the case that clusters are allowed to be non-disjoint. This can help to satisfy the connectivity constraints.

Our main result is an $O(1)$-approximation algorithm for the disjoint connected $k$-center and $k$-diameter problem for Euclidean spaces of low dimension (constant $d$) and for metrics with constant doubling dimension. For general metrics, we get an $O(\log^2k)$-approximation. Our algorithms work by computing a non-disjoint connected clustering first and transforming it into a disjoint connected clustering.

We complement these upper bounds by several upper and lower bounds for variations and special cases of the model.
\end{abstract}

\newpage
\pagenumbering{arabic}
\setcounter{page}{1}

\section{Introduction}

Clustering problems occur in a wide range of application domains. Because of the general importance and interesting combinatorial properties, well-known $k$-clustering problems like $k$-center, $k$-median, and $k$-means have also been vastly studied in theory. These problems are NP-hard and APX-hard, but many constant-factor approximation algorithms for them are known.
All $k$-clustering problems ask to partition a set of points (usually in a general metric space or in Euclidean space) into $k$ clusters, often by picking $k$ centers and assigning every point to its closest center. The clusters are then evaluated based on the distances between the points and their corresponding centers. For example in the case of $k$-center, the objective is to minimize the maximum distance between any point and its closest center.

In applications, clustering problems are often subject to side constraints. Consequently, clustering with side constraints has also become a thriving topic for designing approximation algorithms. Probably the most known example is clustering with capacities where the number of points in a cluster is limited. Notice how this constraint prevents us from assigning points to their closest center because there might not be enough space. So, for example, uniform capacitated (center-based) clustering consists of finding $k$ centers and an assignment of points to those centers such that every center gets at most $U$ points (and then evaluating the desired objective). Finding a constant factor approximation for uniform capacitated $k$-median clustering is a long standing open problem. Other constraints that have been studied are for example lower bounds (here, a cluster has to have a certain minimum number of points, so it may be beneficial to open less than $k$ clusters) and clustering with outliers (here we are allowed $k+z$ clusters, but $z$ of them have to be singletons, i.e. outliers). There are also results on constraints that restrict the choice of centers, for example by demanding that the centers satisfy a given matroid constraint. Among the newer clustering problems with constraints are those that evolve around aspects of fairness. These constraints are typically more complex and can either be point-based or center-based. Each constrained clustering problem, old or new, comes with a unique combinatorial structure, giving rise to a plethora of insights on designing approximation algorithms.

\newcommand{\drsdebuggrid}[4]{
  \draw[gray!50!white, step=0.5,thin] (#1, #2) grid (#3, #4);
  \foreach \x in {#1,...,#3} {
    \node [below, font={\sffamily\footnotesize}, gray] at (\x, #2) {\x};
  }
  \foreach \y in {#2,...,#4} {
    \node [left, font={\sffamily\footnotesize}, gray] at (#1, \y) {\y};
  }
}

In this paper, we study a constraint that stems from the area of sea level geodesy but which is also of interest for other domains (discussed briefly below). For the application that motivated our work, consider the left picture in Figure~\ref{fig:maps}.
\begin{figure}
\centering
\begin{tikzpicture}
\node at (0,0) {\includegraphics[width=0.43\textwidth]{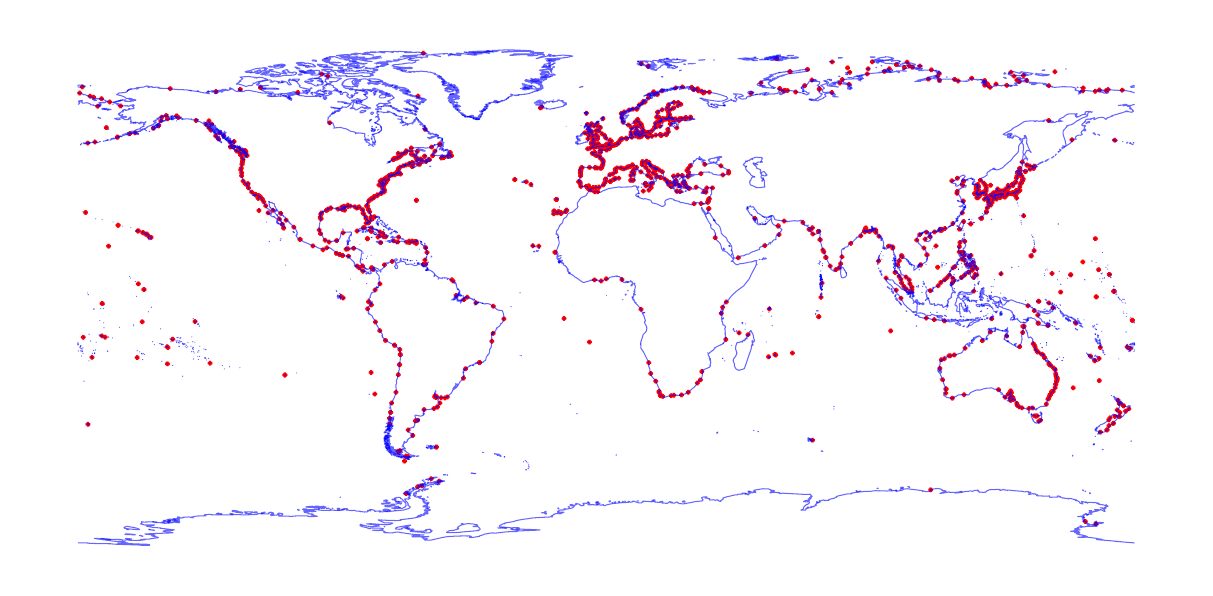}};
\node at (8,0) {\includegraphics[width=0.43\textwidth]{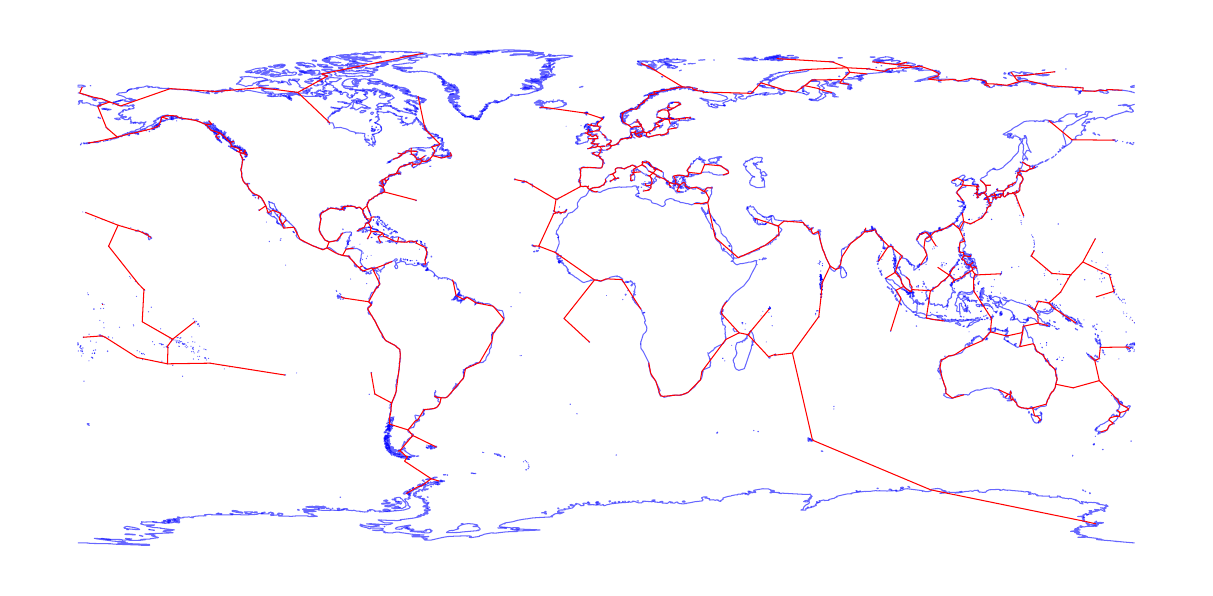}};
\node [fill=white,draw=black, inner sep=0cm] at (6,-1) {\includegraphics[width=3.2cm]{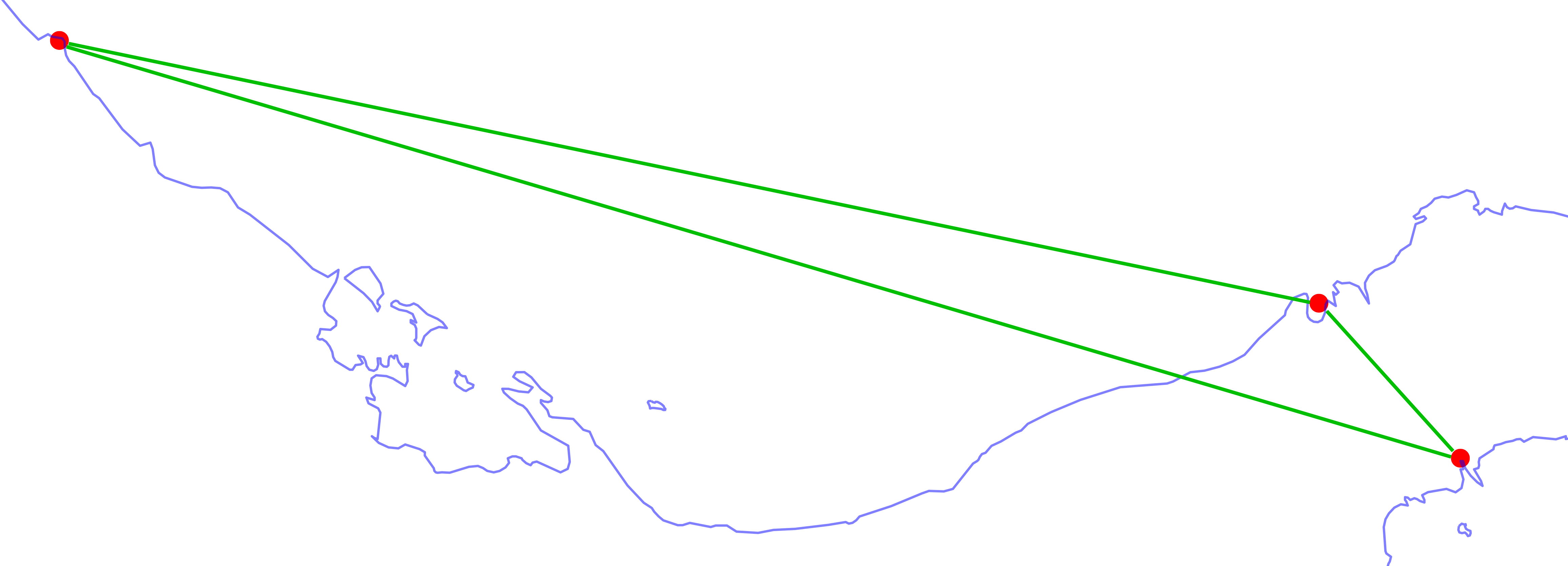}};
\coordinate (p1) at (6.67,0.173) {};
\coordinate (p2) at (6.72,0.163) {};
\coordinate (p3) at (6.726,0.156) {};
\draw[line width=0.05mm] ($(p1)+(-0.01,0.01)$) rectangle ($(p3)+(0.01,-0.01)$);
\coordinate (r1) at (4.397,-0.42);
\coordinate (r2) at (7.6,-0.42);
\draw [gray] ($(p1)+(-0.01,0.01)$) to (r1);
\draw [gray] ($(p3)+(0.01,0.027)$) to (r2);
\node at (6,-0.7) {\fontsize{4}{5}\selectfont $27$ };
\node at (7.41,-1.15) {\fontsize{1}{5}\selectfont $119$};
\node at (6,-1.1) {\fontsize{4}{5}\selectfont $114$ };
\end{tikzpicture}
\caption{Gauge stations around the globe, with station location data from PSMSL (\href{http://www.psmsl.org/data/obtaining/}{http://www.psmsl.org/data/obtaining/}), plotted onto the map from the Natural Earth data set (\href{https://www.naturalearthdata.com/downloads/10m-physical-vectors/10m-coastline/}{https://www.naturalearthdata.com/downloads/10m-physical-vectors/10m-coastline/}). Highlighted are three stations in Central America, and the numbers are Fréchet distances computed on the curves defined by sea levels between 1953 and 1968. \label{fig:maps}}
\end{figure}
 We see the location of tide gauge stations around the globe from the PSMSL data set~\cite{PSMSL-data,PSMSL-paper}. 
At every station, sea level heights have been collected over the years, constituting monthly time series. These records can be used to reconstruct regional or global mean sea levels. However, the tide gauges have usually been constructed for practical purposes and not for sea level science. As a result, they are unevenly distributed over the globe. One way out of this is to replace clusters of tide gauges by representative records to thin out the data set. 
Our general goal is therefore to cluster the tide gauges into a given number $k$ of clusters. However, the objective is not based on the gauge stations' geographic distance but on the time series. 
 We wish to combine gauge stations with similar time series into one, i.e., when we cluster, we want to find clusters where the center's time series is similar to the records collected at the tide gauges represented by that center.
We can model the distance between time series by a metric distance measure for times series or curves (like the Fréchet distance). As the objective we pick $k$-center, so we want to minimize the maximum distance between the center and the points that are replaced by it.
 Now we get to the complication: The gauge stations are \emph{also} points on the map. We do not want to have points in the same cluster that are geographically very far away.

 It is not immediately clear how to best model this scenario. We could resort to bicriteria approximation and look for solutions where both the time series of points in a cluster are similar and the radius of clusters is small, by either looking at the Pareto front or weighting the two objectives. Alternatively, we could fix a threshold and limit the geographic distance between centers and points, i.e., demand that a point $x$ can only be assigned to center $c$ if its geographic distance is at most some $T$. Both modelings have the drawback that they really only capture the distance on the map, while in reality, we would like to have somewhat coherent  clusters that correspond to non-overlapping areas on the map. Indeed, we might be fine with having points of large geographic distance in the same cluster if all points \lq between\rq\ them are also in the same cluster (i.e., that larger area of the sea behaves very similar with respect to the gauge station measurements).

The modeling that we study incorporates this via a preprocessing step. We assume that the points have been preprocessed such that we get a connectivity graph like shown on the right in Figure~\ref{fig:maps}. The graph on the map was computed by finding a minimum spanning tree of the points, but it could be computed in other ways, too. The important part is that it captures a neighborhood structure. To model coherence, we now demand that clusters are \emph{connected} in this graph. Figure~\ref{fig:linegraphexample} gives an example.

 \begin{problem}
  In a connected $k$-clustering problem, we are given points $V$, a metric $d$ on $V$, a number $k$, and an unweighted and undirected connectivity graph $G=(V,E)$. A feasible solution is a partitioning of $V$ into $k$ clusters $C_1,\ldots,C_k$ which satisfies that for every $i\in\{1,\ldots,k\}$ the subgraph of $G$ induced by $C_i$ is \emph{connected}.
 \end{problem}

For the connected $k$-center problem, a solution also contains centers $c_1,\ldots,c_k$ corresponding to the clusters $C_1,\ldots,C_k$ and the objective is to minimize the maximum radius $\max_{i \in [k],x \in C_i} d(x,c_i)$.
For the connected $k$-diameter problem the objective is to minimize the maximum diameter $\max_{i\in[k]}\max_{x,y \in C_i} d(x, y)$. It is easy to see that the connected $k$-clustering problem generalizes the classic $k$-center and $k$-diameter problems when choosing the connectivity graph $G$ as a complete graph.

Interestingly, the connected $k$-center problem was independently defined in an earlier paper by Ge et al.~\cite{GeEGHBB08} (previously unknown to us. We thank the anonymous reviewer who pointed us to this reference.) In that paper, connected clustering is motivated in the context of  applications where both attribute and relationship data is present. It is applied to scenarios of community detection and gene clustering, showing the wide applicability of the modeling. We discuss their work further in the related work section.

\subsection*{Disjoint vs non-disjoint clusters} Notice that we demand that the $C_i$ are \emph{disjoint}. For some clustering problems with constraints the objective value can be decreased when we are allowed to assign points to more than one cluster: For example, lower bounds are easier to satisfy when points can be reused. The same is true for connected clustering: It is easier to satisfy connectivity when we can put important points into multiple clusters. For our application, we want to have disjoint clusters, but we still study the variation for completeness and also since it allows for better approximation algorithms that can be at least tested for their usefulness in the application (e.g., leaving it to the user to resolve overlaps). Notice that in Figure~\ref{fig:linegraphexample}, allowing non-disjoint clusters enables the solution $\{c,d\}$ with clusters $\{a,b,c,d\}$, $\{c,d,e,f\}$ which has cost $1$.

\begin{figure}
 \centering
\begin{tikzpicture}
 \begin{scope}

 \node [circle,draw,inner sep=0cm, minimum width=0.2cm,label=below:{$a$}] (a) at (0,0) {};
 \node [circle,draw,inner sep=0cm, minimum width=0.2cm,label=below:{$b$}] (b) at (1,0) {};
 \node [circle,draw,inner sep=0cm, minimum width=0.2cm,label=below:{$c$}] (c) at (2,0) {};
 \node [circle,draw,inner sep=0cm, minimum width=0.2cm,label=below:{$d$}] (d) at (3,0) {};
 \node [circle,draw,inner sep=0cm, minimum width=0.2cm,label=below:{$e$}] (e) at (4,0) {};
 \node [circle,draw,inner sep=0cm, minimum width=0.2cm,label=below:{$f$}] (f) at (5,0) {};
 \draw [thick,dashed] (a) -- (b) -- (c) -- (d) -- (e) -- (f);
 \draw [thick,bend left=30] (a) to (d);
 \draw [thick,bend left=30] (b) to (d);
 \draw [thick,bend left=30] (c) to (e);
 \draw [thick,bend left=30] (c) to (f);
 \draw [thick, bend right = 30] (c) to (d);
 \end{scope}

 \begin{scope}[xshift=-5cm]
 \node [circle,draw,inner sep=0cm, minimum width=0.2cm,label=above:{$a$}] (a) at (0,0) {};
 \node [circle,draw,inner sep=0cm, minimum width=0.2cm,label=left:{$b$}] (b) at (1,-1) {};
 \node [circle,draw,inner sep=0cm, minimum width=0.2cm,label=above:{$c$}] (c) at (2,0) {};
 \node [circle,draw,inner sep=0cm, minimum width=0.2cm,label=above:{$d$}] (d) at (1,0) {};
 \node [circle,draw,inner sep=0cm, minimum width=0.2cm,label=right:{$e$}] (e) at (2,-1) {};
 \node [circle,draw,inner sep=0cm, minimum width=0.2cm,label=above:{$f$}] (f) at (3,0) {};

 \draw [thick,dashed] (a) -- (b) -- (c) -- (d) -- (e) -- (f);
 \draw [very thick] (a) to (d);
 \draw [very thick] (b) to (d);
 \draw [very thick] (c) to (e);
 \draw [very thick] (c) to (f);
 \draw [thick, bend right = 30] (c) to (d);
 \end{scope}

 \end{tikzpicture}
 \caption{An example input.
 The solid edges form the metric: Vertices connected by a solid edge have distance $1$ and all other distances are $2$. The dashed edges form the connectivity graph.
 Both pictures show the same graph. 
 The optimal $k$-center solution with centers $\{c,d\}$ and clusters $\{a,b,d\}$ and $\{c,e,f\}$ is not connected. Any optimal (disjoint) connected $k$-center solution has radius $2$.
 \label{fig:linegraphexample}}
\end{figure}
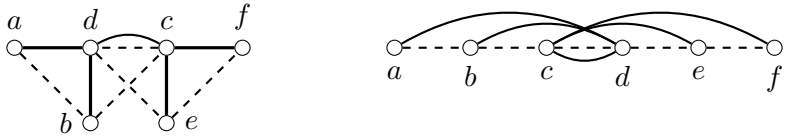

\begin{definition}
 We distinguish between connected $k$-clustering \emph{with disjoint clusters} and \emph{with non-disjoint} clusters, referring to whether the clusters $C_i$ have to be pairwise disjoint or not.
\end{definition}

Finally, we observe that in our application the connectivity graph is not necessarily arbitrary. Depending on the way that we build the graph, it could be a tree (the minimum spanning tree) or even a line (if we follow the coast line). Thus, we are interested in the problem on restricted graph classes as well. In Figure~\ref{fig:linegraphexample}, the connectivity graph is a line.

\paragraph{Results.}
Our main result is an approximation algorithm that works for both the disjoint connected $k$-center problem and the disjoint connected $k$-diameter problem for general connectivity graphs $G$. For general metrics, the algorithm computes an $O(\log^2 k)$-approximation. If the metric has bounded doubling dimension $D$, the approximation ratio improves to $O(2^{3 \cdot D})$, and for $d$-dimensional Euclidean spaces with $L_p$-norm to $O(d^{2+1/p})$.

In Section~\ref{sec:resultssummary} we set this result into context by discussing different classes of restricted connectivity graphs (mostly lines, stars and trees) and also the variant of non-disjoint connected clustering. 
The main result of our work on restricted graph classes is an exact algorithm for trees. We have now learned that Ge et al. have already developed a quite similar algorithm, and thus we have moved our results on trees to Appendix~\ref{chap:treeopt}. We review what is known about restricted graph classes in Section~\ref{sec:resultssummary}.
Table~\ref{table} in that section gives an overview of the smaller results that we obtain and of what was known previously.

\paragraph{Techniques.}
For the general result, we first compute a non-disjoint clustering. Then we develop a method using a concept of a layered partitioning (see Definition~\ref{def:wellseparated}) to make the clusters disjoint. We show how to obtain such a partitioning for different metrics. Both steps are novel and form the main contribution of this paper.
We use dynamic programming for obtaining optimal solutions on trees (this was already done similarly by Ge et al.~\cite{GeEGHBB08} in independent work). We also demonstrate how to solve the line case by greedy algorithms. Then we derive some hardness results. Most of them are short reductions, but the worst-case instance in Sections~\ref{sec:WorstCaseInstance} and~\ref{sec:WorstCaseInstancePartII} and the hardness proof in Section~\ref{lower:bound:three} require more complicated constructions.

\subsection*{Related work.}
The $k$-center problem and the $k$-diameter problem are both NP-hard to approximate  better than by a factor of 2 (see~\cite{HN79,H84} for $k$-center, $k$-diameter follows along the same lines).
There are two popular $2$-approximation algorithms for $k$-center which both also work for $k$-diameter with the same approximation guarantee~\cite{G85,HS86}.

\paragraph{Related work on connected clustering.}

The connected $k$-center problem with disjoint clusters has been introduced and studied by Ge et al.~\cite{GeEGHBB08}\footnote{We thank an anonymous reviewer for pointing us to this reference}. Besides other results, Ge et al.\ present a greedy algorithm for the problem and claim that it computes a 6-approximation. In Appendix~\ref{sec:Counterexample} we present an example showing that this greedy algorithm actually only obtains an $\Omega(k)$-approximation. The greedy algorithm is based on the approach of transforming a non-disjoint clustering into a disjoint one. In this transformation, it does not change the centers, i.e., it uses the given centers of the non-disjoint clustering also as centers for the disjoint clustering. In Section~\ref{sec:WorstCaseInstance} we prove a lower bound showing that no algorithm based on transforming a non-disjoint clustering into a disjoint one with the same centers can compute an $O(1)$-approximation. Hence, without fundamental changes of the algorithm, no $O(1)$-approximation can be obtained. In Section~\ref{sec:WorstCaseInstancePartII} we even show that in general the optimal non-disjoint clustering can be better than the optimal disjoint clustering by a factor of $\Omega(\log\log{k})$. Hence, if one uses only the radius of an optimal non-disjoint clustering as a lower bound for the radius of an optimal disjoint clustering, one cannot show a better approximation factor than $\Omega(\log\log{k})$. To the best of our knowledge, no other approximation algorithms with provable guarantees for the connected $k$-center or $k$-diameter problem are known.

Ge et al.\ introduce the connected $k$-center problem to model clustering problems where both attribute and relationship data is present. They perform experiments in the context of gene clustering and community detection and demonstrate that for both these applications modelling them as connected clustering problems leads to superior results compared to standard clustering formulations without connectivity constraint. For community detection for example, they construct datasets from DBLP\footnote{\url{https://dblp.org/}} where researchers are supposed to be clustered according to their main research area. Based on keyword frequencies they defined a distance measure for the researchers. At the same time, the coauthor network can be used as a connectivity graph. The advantage of connected clustering compared to traditional models is that it naturally takes into account both the distance measure and the coauthor network. For their experiments, Ge et al.\ develop a heuristic called NetScan for the connected $k$-center problem with disjoint clusters, which is reminiscent of the $k$-means method, and efficient on large datasets. In their experiments, the outcomes of this heuristic were significantly better than the outcomes of state-of-the-art clustering algorithms that take into account either only the distance measure or only the coauthor network. The work of Ge et al.\ has attracted some attention and it is cited in many other articles on community detection and related subjects.

Furthermore, Ge et al. show that already for $k=2$, the connected $k$-center problem with disjoint clusters is NP-hard. They also argue that it is even NP-hard to obtain a $(2-\epsilon)$-approximation for any $\epsilon>0$. Additionally they give an algorithm based on dynamic programming with running time $O(n^2\log{n})$ that solves the connected $k$-center problem with disjoint clusters optimally when the connectivity graph is a tree, similar as we did (later) as it is described in Appendix~\ref{chap:treeopt}.

Gupta et al.~\cite{GuptaPS11} study the connected $k$-median and $k$-means problem and prove upper and lower bounds on their approximability. Related to our motivation, Liao and Peng~\cite{LiaoP12} consider the connected $k$-means problem to model clustering of spatial data with a geographic constraint. They develop a local-search based heuristic and conduct an experimental evaluation.

\paragraph{Related work on clustering with side constraints.}

It is beyond the scope of this paragraph to list all constant-factor approximations for clustering with side constraints. Constraints that aim at fairness, diversity or non over-representation have gained a lot of attention, e.g., see~\cite{AEKM19,BIOSVW19,BCFN19,BGKKRS019,CFLM19,CKLV17,CKR20,KAM19,LiYZ10,TOG21}.
Lower bounds on cluster sizes are a classical topic that received a renewed interest because they can be used to model a very mild form anonymization~\cite{AS12-lowerbounded,AS16,AS21,RS18,S10}. In particular, \cite{AS21} explores the idea of adding points to more than one cluster to make it easier to satisfy lower bounds.
Clustering with \emph{upper} bounds (capacities) on the cluster sizes has always been popular and is still producing many interesting results~\cite{ABCGMS15,ASS17,CHK12,KS00}. Outliers\footnote{Clustering with outliers allows to ignore $z$ points from the input point set. Formulated as a constraint, it allows $k+z$ clusters instead of $k$ clusters under the constraint that $z$ clusters are singletons.} are also a popular topic~\cite{CGK20,CKMN01,CN19,HJLW18,KLS18,McCutchenK08}. There are many more interesting results on clustering with constraints, but we are not aware of any approximation algorithm for connected clustering as we defined it above. The variant where points can only be assigned to a center if they are within a certain distance of it was studied in~\cite{GGS16} in combination with $k$-median, the problem is named \emph{local $k$-median}. Since the problem captures set cover, it does not allow for a constant-factor approximation, and a bicriteria approximation for the problem is given in~\cite{GGS16}.

When the clustering problem is $k$-means in Euclidean space and the constraint can be expressed as a constraint on the \emph{allowed partitionings} of the input point set into clusters (this is true for capacitated clustering and also for many variants of fairness constraints), then one can use a generic framework described in~\cite{DX15,BJK18}. 

There is also a more generic approach to solving clustering with constraints which restrict the choice of \emph{centers}. If the constraint can be formulated by a matroid constraint (the set of centers has to be independent in a suitably defined matroid), there are generic constant-factor approximations by Chen et al.~\cite{CLLW16} for the $k$-center case and by Krishnaswamy et al.~\cite{KKNS15} for the $k$-median case. Even so, there are problem-tailored approximation algorithms for specific cases to improve upon the generic runtime, for example for $k$-center with a center-based fairness criterion~\cite{CKR20,KAM19}.

In \cite{CGLMPS11}, Cygan et al.\ consider a variant of the facility location problem where an instance can have two or more different cost functions on the same set of facilities and clients, and the goal is to find a solution that is good for all cost functions. They give an $\tilde{O}(k^{1-1/h})$-approximation algorithm, where $h$ is the number of cost functions and $\tilde{O}$ suppresses logarithmic factors.

\paragraph{Outline.}
In Section~\ref{sec:resultssummary} we give an overview of our results and prove the first results for non-disjoint clusterings. In Section~\ref{chap:graph} we prove our main upper bounds for general connectivity graphs and disjoint clusterings, while Section~\ref{sec:LowerBounds} contains our lower bounds. We conclude the main part of this article by considering the case that the connectivity graph is a line in Section~\ref{sec:lineproofs}. In Appendix~\ref{sec:Counterexample} we present a counterexample to the approximation factor of $6$ for the approximation algorithm by Ge et al.~\cite{GeEGHBB08}. In Appendix~\ref{chap:treeopt}, we consider the case that the connectivity graph is a tree and we present a dynamic programming approach. This result has already been obtained by Ge et al.~\cite{GeEGHBB08}.
\section{Overview and First Results}\label{sec:resultssummary}

In this section, we give some intuition about connected clustering, review results about restricted graph classes and briefly discuss the proof of our main result for general graph classes.
We use the following well-known fact for $k$-center/$k$-diameter approximation first used by Hochbaum and Shmoys~\cite{HS86}: For the $k$-diameter or $k$-center problem (connected or not), the value of the cost function is always equal to one of the at most $n^2$ different distances between two points in $V$ where $n=|V|$. This is true because it is either the distance between two points in the same cluster ($k$-diameter) or it is the distance between a point and its center ($k$-center). Thus, a standard scheme to follow is to sort these distances in time $O(n^2 \log n)$ and then search for the optimum value by binary search. The problem then reduces to finding a subroutine for the following task.

\begin{problem}\label{p:subroutine}
If there is a solution which costs $r$ for a given~$r$, find a solution that costs at most $\alpha \cdot r$.
Otherwise, report that $r$ is too small.
\end{problem}

An algorithm that solves this task can easily be turned into an $\alpha$-approximation by searching for the smallest $r$ for which the algorithm returns a solution.
The running time of the resulting algorithm is $O(n^2 \log n)$ for the preprocessing plus $O(\log n)$ times the running time of the subroutine.

\begin{table}
    \centering
    \begin{adjustbox}{width=1\textwidth}
    \small
        \begin{tabular}{l|c|c|c|c}
        \multirow{2}{*}{\diagbox{Restriction}{Objective}}           & \multicolumn{2}{c|}{$k$-Center}                                                                                & \multicolumn{2}{c}{$k$-Diameter}\\
        \cline{2-5}
        & disjoint & non-disjoint & disjoint & non-disjoint \\
        \hline
        $G$ is a line          & \multirow{3}{*}{\begin{tabular}[c]{@{}c@{}}$1$\\[-0.2cm]\tiny{\mbox{Ge et al.~\cite{GeEGHBB08}}, also see Appendix~\ref{chap:treeopt}}\end{tabular}}          & \begin{tabular}[c]{@{}c@{}}$1$\\[-0.2cm]\tiny{\mbox{Cor. \ref{thm:line:all}}}\end{tabular}                          & \multicolumn{2}{c}{\begin{tabular}[c]{@{}c@{}}$1$\\[-0.2cm]\tiny{\mbox{Lem. \ref{thm:opt_diam_line}}}\end{tabular}} \\
        \cline{1-1}\cline{3-5}
        $G$ is a star / tree   &                              &
        & \begin{tabular}[c]{@{}c@{}}$[2, 2]$\\[-0.2cm]\tiny{\mbox{Lem.  \ref{thm:tree_lowerbound_diameter}, Thm. \ref{thm:opt_kcenter_tree} }}\end{tabular} & 
        \\
        \cline{1-2}\cline{4-4}
        \small{Doubling dimension $D$} & {\begin{tabular}[c]{@{}c@{}}$ O(2^{3D})$\\[-0.2cm]\tiny{\mbox{Thm. \ref{thm:doubling_dimension}}}\end{tabular} }           & \begin{tabular}[c]{@{}c@{}}$[2, 2]$\\[-0.2cm]\tiny{\mbox{Lem. \ref{thm:lower_bound_non-disjoint_k-center}, Lem. \ref{lem:2-approx_non-disjoint}}}\end{tabular}  & {\begin{tabular}[c]{@{}c@{}}$ O(2^{3D})$\\[-0.2cm]\mbox{\tiny{Thm. \ref{thm:doubling_dimension}}}\end{tabular} } &     \begin{tabular}[c]{@{}c@{}}$[2, 2]$\\[-0.2cm]\mbox{\tiny{Lem.  \ref{thm:tree_lowerbound_diameter_non-disjoint}, Lem. \ref{lem:2-approx_non-disjoint}}}\end{tabular}                       \\
        \cline{1-2}\cline{4-4}
        \small{$L_p$ metric in dimension $d$} & {\begin{tabular}[c]{@{}c@{}}$ O(d^{2+1/p})$\\[-0.2cm]\mbox{\tiny{Thm. \ref{thm:euclidean_approx}}}\end{tabular} }           & & {\begin{tabular}[c]{@{}c@{}}$ O(d^{2+1/p})$\\[-0.2cm]\mbox{\tiny{Thm. \ref{thm:euclidean_approx}}}\end{tabular} } &                            \\
        \cline{1-2}\cline{4-4}
        No Restrictions & {\begin{tabular}[c]{@{}c@{}}$[3,O(\log^2 k)]$\\[-0.2cm]\mbox{\tiny{Thm. \ref{thm:lower_bound_4-center}, Thm. \ref{thm:log-approx_disjoint}}}\end{tabular} }           & & {\begin{tabular}[c]{@{}c@{}}$[2, O(\log^2 k)]$\\[-0.2cm]\mbox{\tiny{Lem. \ref{thm:tree_lowerbound_diameter}, Thm. \ref{thm:log-approx_disjoint}}}\end{tabular} } &

        \end{tabular}
    \end{adjustbox}
\caption{An overview of the bounds shown in this paper and the literature for connected $k$-clustering. The notation $[\ell, u]$ stands for a lower bound $\ell$ and an upper bound $u$ on the approximation factor (achievable in polynomial time and assuming $P\neq NP$). \label{table}}
\end{table}

\paragraph{Lines, stars and trees.}
Connected $k$-clustering demands that the clusters are connected in a given connectivity graph $G$. How tricky is this condition? Maybe it can actually \emph{help} to solve the problem?
This is true if $G$ is very simple, i.e., a line.

\begin{restatable}{corollary}{thmLineAll}
    \label{thm:line:all}
    When the connectivity graph $G$ is a line graph, then the connected $k$-center problem and the connected $k$-diameter problem can be solved optimally in time $O(n^2\log n)$ both with disjoint and non-disjoint clusters. This is true even if the distances are not a metric.
\end{restatable}
\begin{proof}
We only show how to solve the connected $k$-center problem with non-disjoint clusters. The full proof can be found in Section~\ref{sec:lineproofs}. 
The line graph $G$ is defined by vertices $V=\{v_1, v_2, ..., v_n\}$ and edges $E=\{\{v_i, v_{i+1}\}\mid i\in \{1,\ldots,n-1\}\}$.
Assume that $r$ is given.

 Notice that any connected cluster is a subpath of $G$.
   We start by precomputing for every $v_i$ how far a cluster with center at $v_i$ can stretch to the left and right: Let $a_i$ be the smallest $\ell$ such that $d(v_j,v_{j'})\le r$ for all $j,j' \in \{\ell,\ldots,i\}$ and let $b_i$ be the largest $\ell$ such that $d(v_j,v_{j'})\le r$ for all $j,j' \in \{i,\ldots,\ell\}$. 
	We can compute all $a_i$ and all $b_i$ in time $O(n^2)$.
   Now we cut the line into clusters. We start by finding an index~$i$ with $a_i=1$ for which $b_i$ is as large as possible because we have to cover the first vertex and want to cover as many other vertices as possible. We place a center at $v_i$ and know that all vertices until $v_{b_i}$ are covered by the cluster.
Now we know that the next cluster has to contain $v_{b_i+1}$, so we search for an $i'$ 
 which satisfies $b_i+1 \in \{a_{i'}, \ldots,b_{i'}\}$, if there are multiple, we take the one with maximum $b_{i'}$. This finds the center which covers $v_{b_i+1}$ and the largest number of additional vertices. 
We place a center at $v_{i'}$. It may be that $i' < i $ as in Figure~\ref{fig:linegraphexample}) and thus the clusters have to overlap (recall that we are in the non-disjoint case).
The process is iterated until $v_n$ is covered. 
   If the number of clusters is more than $k$, we report that $r$ was too small, otherwise, we report the clustering. This way we solve Problem~\ref{p:subroutine} for~$\alpha=1$ in time $O(n^2)$.
\end{proof}

For trees, $k$-center and $k$-diameter differ. Surprisingly, the connected $k$-diameter problem is already NP-hard if $G$ is a star. We prove the following lemma by a reduction from the uniform minimum multicut problem on stars in Section~\ref{sec:hardness:star:diameter}.

\begin{restatable}{lemma}{lemLowerBound}
\label{thm:tree_lowerbound_diameter}
Let~$\epsilon>0$. Assuming P~$\neq$~NP, there is no $(2-\epsilon)$-approximation algorithm for the connected k-diameter problem with disjoint clusters even if $G$ is a star.
\end{restatable}

Notice how the connected $k$-diameter problem with $G$ being a star is thus very different from the $k$-diameter problem where the \emph{metric} is given by a graph metric that is a star. The latter problem can be solved optimally by sorting the edges by weight and then deleting the $k-1$ most expensive edges to form $k$ connected components which form an optimal clustering. Say we have distances $d(e_1) \ge d(e_2) \ge \ldots \ge d(e_n)$, then this optimal clustering has cost $d(e_{k})+d(e_{k+1})$. However, any clustering that keeps an edge from $\{e_1,\ldots,e_{k-1}\}$ costs at least $d(e_{k+1})+d(e_{k-1}) \ge d(e_{k})+d(e_{k+1})$ since it deletes at most $k-1$ edges.

Ge et al.~\cite{GeEGHBB08} show that the connected $k$-center problem is still solvable optimally for trees by dynamic programming. In Appendix~\ref{chap:treeopt} we explain a similar derivation of a dynamic program for trees.

\begin{restatable}{theorem}{thmTreeOpt}[Ge et al.~\cite{GeEGHBB08}]
    \label{thm:opt_kcenter_tree}
    When the connectivity graph $G$ is a tree, then the connected $k$-center problem with disjoint clusters can be solved optimally in time $O(n^2\log n)$. This is true even if the distances are not a metric.
\end{restatable}

It follows immediately that the (metric) connected $k$-diameter problem with disjoint clusters on trees can be $2$-approximated by the same algorithm. This is tight because of Lemma~\ref{thm:tree_lowerbound_diameter}.

\paragraph{General $G$, non-disjoint clusters.} 
The connected $k$-center and $k$-diameter problems with non-disjoint clusters behave similarly to the unconstrained versions. On the positive side, there is a $2$-approximation; on the negative side, it is NP-hard to approximate these problems better than $2$. In contrast to the case of disjoint clusters, APX-hardness starts with stars for \emph{both} $k$-center and $k$-diameter.
We show this via reductions from clique cover and set cover in Section~\ref{sec:non-disjoint}.

\begin{restatable}{corollary}{thmNonDisjointBoth}
    \label{thm:nondisjoint:both}
    Let $\epsilon>0$.
    Assuming P~$\neq$~NP,  there is no $(2-\epsilon)$-approximation algorithm for the connected k-diameter problem with non-disjoint clusters, even if $G$ is a star. The same is true for the connected $k$-center problem with non-disjoint clusters.
\end{restatable}

For the positive result, the classical result by Hochbaum and Shmoys~\cite{HS86} can be used. For the unconstrained $k$-center problem, Problem~\ref{p:subroutine} for $\alpha=2$ can be solved as follows: Given input $V$, $k$, and a radius $r$, one picks an arbitrary point~$x\in V$ and puts all nodes within distance~$2r$ of~$x$ into one cluster. When~$r$ is at least the radius of the optimal $k$-clustering, this cluster will contain all nodes that are in the same optimal cluster as~$x$. The cluster is then removed from~$V$ and the process is repeated until all nodes are covered. If the number of clusters is at most $k$, the solution is returned, otherwise, it is reported that $r$ was too small.

This algorithm can easily be adapted to the connected $k$-center problem with non-disjoint clusters by the following observation: Let~$x$ and~$y$ be two nodes from the same optimal cluster with center~$c$ and radius~$r$. Then $x$ and~$y$ are connected in the connectivity graph by a path that contains only nodes within distance~$2r$ from~$x$ and~$y$. 
So the algorithm is:
When a node~$x$ is selected, put all nodes into a cluster that have distance at most $2r$ from~$x$ and are reachable from~$x$ in the connectivity graph via a path on which all nodes have a distance of at most~$2r$ from~$x$. This set can be determined by the BFS-type algorithm \texttt{ComputeCluster} (see Algorithm~\ref{alg:compute_Cluster:intro} with $R=2r$). Say the resulting cluster is $T$. Do not remove $T$ from $G$ but only mark all nodes in $T$ as covered.
As long as there are uncovered nodes, pick an arbitrary such node and form a cluster of radius $2r$ around it (in general this cluster will also contain nodes that are already covered). This will result in at most $k$ connected clusters with radius~$2r$ if $r$ is at least the radius of an optimal connected $k$-clustering. We call this algorithm \texttt{GreedyClustering}.

\begin{algorithm}
{\small
   \caption{\textsc{ComputeCluster}$(G,M,R,c)$}\label{alg:compute_Cluster:intro}

   \KwIn{points $V$, graph $G=(V,E)$, metric $M=(V,d)$, radius $R$, node $c\in V$}
      $T\gets \{c\}$\;
      $N \gets \{u\in V\setminus T\mid \exists v\in T, (v,u)\in E: d(u, c)\le R\}$\;

      \While{$N\neq \emptyset$}
      {
         $T \gets T \cup N$\;
         $N \gets \{u\in V\setminus T\mid \exists v\in T, (v,u)\in E: d(u, c)\le R\}$\;
      }
   \KwOut{cluster $T$}
}
\end{algorithm}

The same algorithm works for the connected $k$-diameter problem when \texttt{ComputeCluster} is evoked with $R=r$ (not $2r$) if $r$ is at least the optimal diameter. By adding all points in distance $r$ to the cluster of the chosen center $x$, it is ensured that the optimum cluster is added if $r$ is at least the optimum value (since the distance between two points is then at most $r$). Furthermore, the resulting cluster has diameter at most $2r$ by the triangle inequality.

\begin{lemma}
   \label{lem:2-approx_non-disjoint}
   There exists a $2$-approximation algorithm for the connected $k$-center problem with non-disjoint clusters and also for the connected $k$-diameter problem with non-disjoint clusters.
\end{lemma}

\paragraph{General case.}
The disjoint case for general connectivity graphs is more challenging. We start with algorithm \texttt{GreedyClustering} from the previous paragraph on the non-disjoint case. Notice that in general, the output of this algorithm is not node-disjoint. We could opt to delete the nodes in $T$ computed by Algorithm~\ref{alg:compute_Cluster:intro} to enforce disjointness, however, the problem is this:
\begin{figure}
\centering
\begin{tikzpicture}[scale=1]
 \node [circle,draw,inner sep=0cm, minimum width=0.45cm] (x) at (0,0) {$x$};
 \node [circle,draw,inner sep=0cm, minimum width=0.45cm] (y) at (0,1) {$u$};
 \node [circle,draw,inner sep=0cm, minimum width=0.45cm] (z) at (1.5,-0.5) {$e$};
 \node [circle,draw,inner sep=0cm, minimum width=0.45cm] (a) at (1.5,0.5) {$z$};
 \node [circle,draw,inner sep=0cm, minimum width=0.45cm] (b) at (1.5,1.5) {$c$};
 \draw (y) to node [label=left:{$r$}] {} (x);
 \draw (x) to node [label=above:{$2r$}] {} (a);
 \draw (a) to node [label=right:{$r$}] {} (z);
 \draw (a) to node [label=right:{$r$}] {} (b);
 \node [anchor=west] at (2.5,0.5){\begin{minipage}{8cm}\small The optimal connected $2$-clustering  has centers $x$ and $z$ with clusters $\{x,u\}$ and $\{z,c,e\}$ and a radius of~$r$. 
 The greedy algorithm started with~$x$ forms $\{x,u,z\}$ as the first cluster. After that, only $c$ and $e$ remain. Without $z$, they are not connected anymore and have to go into different clusters.\end{minipage}};
\end{tikzpicture}
\caption{An example where greedy disconnects an optimum cluster.}
\label{fig:BadExampleGreedy:intro}
\end{figure}
The first cluster that the algorithm forms around a vertex~$x$ is guaranteed to be a superset of the optimal cluster that~$x$ is contained in. It might be a strict superset and contain a node that belongs to a different optimal cluster. This node will get removed from~$G$ together with all other nodes in the cluster around~$x$. However, its removal might make the optimal cluster it is contained in unconnected. This is problematic because then $k$ connected clusters might not suffice anymore to cover all points from~$G$ even if we guessed the optimal radius~$r$ correctly. See Figure~\ref{fig:BadExampleGreedy:intro} for an example where this happens.

We first compute the output of the algorithm \texttt{GreedyClustering} for the non-disjoint case from the previous paragraph. In general this is a set of non-disjoint clusters that cover all points. We transform this set into a set of pairwise disjoint clusters that cover all points at the cost of increasing the radius or diameter. This transformation has to be performed very carefully in order to not increase the radius or diameter by too much.

Let $C$ with $|C|\le k$ denote the set of centers around which the non-disjoint clusters have been formed by the algorithm and let~$R$ denote their radius. The following two observations are helpful: (1) When two centers are more than $2R$ apart then their corresponding clusters are disjoint. (2) If a set of centers have pairwise distance at most~$L$ then merging the corresponding clusters results in a single cluster with radius at most $R+L$ and diameter at most $2R+L$. 

If it is possible to partition the centers into groups such that all centers within the same group have a distance of at most~$L$ and all centers from different groups have a distance of more than~$2R$, we could make the clusters disjoint as follows: as long as there are two non-disjoint clusters whose centers are in the same group of the partition, merge them into a single cluster. In the end, the algorithm will return no more than $|C|\le k$ clusters. By isolating some singletons as new clusters, we obtain a solution with exactly $k$ clusters as required without worsening the solution. 
After this, all clusters whose centers are in the same group are disjoint (if not they would have been merged) and clusters whose centers are in different groups are disjoint because their centers are far enough from each other. Hence, such a partition results in a solution with disjoint clusters with radius $R+L$ and diameter $2R+L$. A key idea in our algorithm for the general case is to find such a partition of the centers in~$C$ with small~$L$. However, observe that this is not possible in general. A simple counterexample would be that all centers are equally spaced on a line with distance~$R$ between two consecutive centers. Then all centers have to be in the same group and~$L$ would be $(k-1)R$, resulting in an approximation factor of $\Omega(k)$.

To circumvent this problem, we do not partition all centers from~$C$ at once but we start with a partition of a subset of~$C$ that satisfies the properties above (i.e., centers in the same group have distance at most~$L$, while centers in different groups have a distance of more than~$2R$). We call this the first layer of the partition. Then we remove all centers contained in the first layer from~$C$ and proceed with the remaining centers analogously: Let~$C'$ denote the set of centers not contained in the first layer. We find a partition of a subset of~$C'$ that satisfies the properties above and call this the second layer of the partition. We repeat this process until all points from~$C$ are in some layer. We call such a partition a \emph{well-separated partition}. Figure~\ref{fig:WellSeparatedPartition} shows possible partitions for the example above. 

\begin{figure}
  \centering
  \includegraphics[scale=0.8,page=1]{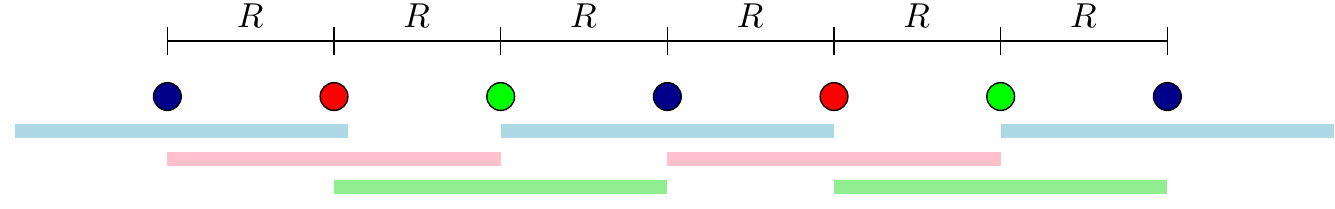}\\
  \centering
  \includegraphics[scale=0.8,page=2]{external/WellSeparated.pdf}
  \caption{We consider an instance with 7 centers on a line where consecutive centers have a distance of~$R$. The top figure shows a well-separated partition of this instance with~$L=0$ and~$\ell=3$ layers. The colors depict the different layers and the colored rectangles depict the clusters of radius~$R$ around these centers. On the blue layer there are, e.g., three groups where each group consists of a single blue center. The bottom figure shows a well-separated partition of the same instance with~$L=R$ and~$\ell=2$. The blue layer contains two groups of two centers each, while the red layer contains two groups, one with two centers and one with only one center.}
  \label{fig:WellSeparatedPartition}
\end{figure}

It is not clear at first glance why a well-separated partition is helpful for obtaining a solution with disjoint clusters. For every layer of the partition, we can use the reasoning above. That is, we merge all non-disjoint clusters whose centers are in the same group to obtain disjoint clusters with radius $R+L$ and diameter $2R+L$. However, a cluster is then only disjoint from all clusters on the same layer but in general not from clusters on other layers (see Figure~\ref{fig:WellSeparatedPartition}). A main ingredient of our algorithm is a non-trivial way to merge clusters on different layers. For this, we add the layers one after another. Consider the case of two layers. The clusters from the first layer are disjoint from each other. We add the clusters of the second layer one after another. For each cluster from the second layer, we first check with which clusters from the first layer it overlaps. If there is more than one, we split the cluster from the second layer into multiple parts and merge the parts with different clusters from the first layer with which they overlap. This is done in such a way that the final result is a set of disjoint connected clusters. We prove with an inductive argument that the radius and diameter of these clusters is $O(\ell\cdot L)$, where~$\ell$ denotes the number of layers of the well-separated partition.

With the discussion above, finding a good approximation algorithm is reduced to finding an efficient algorithm for computing a well-separated partition with small~$L$ and few layers. For general metrics, we present an efficient algorithm that computes a well-separated partition for~$L=O(R\cdot\log{k})$ and~$\ell=O(\log{k})$. This yields a clustering of radius and diameter $O(R\cdot \log^2{k})$. Details can be found in the proof of Theorem~\ref{thm:log-approx_disjoint} below.
We give better results for computing well-separated partitions for $L_p$-metrics and metric spaces with bounded doubling dimension in Theorem~\ref{thm:euclidean_approx}
and Theorem~\ref{thm:doubling_dimension}. 
Overall, we get the following results.

\begin{theorem}
   \label{thm:all results}
   There exists an $O(\log^2{k})$-approximation algorithm for the connected $k$-center problem with disjoint clusters and for the connected $k$-diameter problem with disjoint clusters. The approximation ratio improves
   \begin{itemize}
    \item to $O(2^{3\cdot\mathrm{dim}(M)})$ if the metric space has bounded doubling dimension $\mathrm{dim}(M)$, and
    \item to $O(d^{2+1/p})$ if the distance is an $L_p$-metric in~$\mathbb{R}^d$.
   \end{itemize}
\end{theorem}

In addition, we study how to compute well-separated partitions if the number of clusters is small, particularly when $k=2$. We obtain a $2$-approximation algorithm for the connected $k$-center problem with disjoint clusters in Corollary~\ref{cor:2approx-2freecenters} and a $4$-approximation algorithm for the connected $k$-diameter problem with disjoint clusters in Corollary~\ref{cor:3and4approx-2givencenters}.

It is an intriguing question if better well-separated partitions exist for general metrics and for the special metrics that we have considered. By our framework, better partitions would immediately give rise to better approximation factors.

We show a lower bound of~3 on the approximability of the connected $k$-center problem with disjoint clusters (Theorem~\ref{thm:lower_bound_4-center}). A lower bound of~2 for the approximability of the connected $k$-diameter problem with disjoint clusters follows from the lower bound of~2 for the standard $k$-diameter problem without connectivity constraint. In addition to these lower bounds, we also prove a lower bound of $\Omega(\log\log{k})$ for our algorithmic framework. To be precise, we construct an instance together with a set of $k$ centers~$C$ that could be produced by the algorithm \texttt{GreedyClustering} such that even the optimal disjoint solution with centers~$C$ is worse than the optimal disjoint solution for arbitrary centers by a factor of~$\Omega(\log\log{k})$. Hence, to prove a constant-factor approximation one cannot rely on the centers chosen by \texttt{GreedyClustering}.

\section{Connected Clustering with General Connectivity Graphs}
\label{chap:graph}

In this section, we study the connected $k$-center problem and the connected $k$-diameter problem for general connectivity graphs: given an unweighted graph $G=(V,E)$ and a metric space $M=(V,d)$ with $d:V\times V\to \mathbb{R}$, find $k$ node-disjoint connected subgraphs of $G$ (clusters) that cover all vertices and minimize the maximum radius or diameter of these subgraphs.
To keep the presentation simple, we will focus in the following on the connected $k$-center problem, and later adapt the algorithm and its analysis to the connected $k$-diameter problem. We follow the approach discussed in Section~\ref{sec:resultssummary}. That is we first use the algorithm \texttt{GreedyClustering} to compute a set of non-disjoint clusters and then transform this set into a set of disjoint clusters using a well-separated partition of the metric.

\subsection{Greedy Clustering} 

\begin{algorithm}
   \caption{\textsc{GreedyClustering}$(G,M,r)$}\label{alg:guess_general}
   \KwIn{graph $G=(V, E)$, metric $M=(V, d)$, radius $r$}   
   $C \gets \emptyset$;     \tcp{center nodes}
   $V' \gets V$;     \tcp{uncovered nodes}
   \While{$V'\neq \emptyset$}
   {
      select a node $c\in V'$ and add it to $C$\;
      $T_c\gets$ \textsc{ComputeCluster}$(G,M,r,c)$\;
      $V'\gets V'\setminus T_c$\;
   }
   \KwOut{centers $C$, sets~$T_c$ for all $c\in C$}
\end{algorithm}

We give the pseudocode of \texttt{GreedyClustering} as Algorithm~\ref{alg:guess_general}. In general, the sets~$T_c$ computed by this algorithm are not disjoint but the centers are pairwise distinct.
\begin{lemma}\label{lemma:RadiusWithin2Opt}
Let~$r^*$ denote the radius of an optimal connected $k$-center clustering with non-disjoint clusters.
For~$r\ge 2r^*$, Algorithm~\ref{alg:guess_general} computes a center set~$C$ with~$|C|\le k$.
\end{lemma}
\begin{proof}
Consider a node~$c\in V$ that is chosen as a center by the algorithm and the optimal cluster~$O$ node~$c$ is contained in. This cluster is centered around some node~$c'$ and has a radius of at most~$r^*$. Hence, by the triangle inequality all nodes in~$O$ have a distance of at most~$2r^*$ from~$c$. Also since~$O$ is connected, all nodes in~$O$ are reachable from~$c$. In particular, all nodes in~$O$ are reachable from~$c$ on paths that contain only nodes within distance~$2r^*$ of~$c$. This implies that for~$r\ge 2r^*$, the set~$T_c$ is a superset of the optimal cluster~$O$. Since the centers in Algorithm~\ref{alg:guess_general} are chosen among the uncovered nodes, all chosen centers must be from distinct optimal clusters. This implies that there can be at most~$k$ centers in~$C$.
\end{proof}

In the following, we assume~$r$ to be chosen as the smallest radius for which Algorithm~\ref{alg:guess_general} outputs at most $k$ center nodes. This value can be found by binary search. Since there are no more than $n(n-1)$ distinct distances in the metric $M=(V,d)$, we need only $O(\log{n})$ calls to Algorithm~\ref{alg:guess_general} for this. 

\subsection{Making the Clusters Disjoint}\label{subsec:MakingDisjoint}

In this section we will describe how the clusters returned by Algorithm~\ref{alg:guess_general} can be made pairwise disjoint. Since the radius of all returned clusters is at most~$r$, two clusters formed around centers with a distance of more than~$2r$ are disjoint. On the other hand, if two centers~$c$ and~$c'$ have a distance of at most~$L$ and we merge the two corresponding clusters~$T_c$ and~$T_{c'}$ then we get a new cluster with radius at most~$r+L$ with respect to~$c$ or~$c'$. Hence, two clusters whose centers are close together can be merged without increasing the radius too much.

In order to determine which clusters should be merged, we introduce the notion of a \emph{well-separated} partition, which is a partition of the center set~$C$ with some additional properties.

\begin{definition}
   \label{def:wellseparated}
   Let $M=(C,d)$ be a metric and~$r>0$. An \emph{$r$-well-separated partition} with $\ell\in\mathbb{N}$ layers and with parameters $(h_1,\ldots,h_{\ell})$ is a partition of~$C$ into groups $\{C_{1,1}, \ldots, C_{1,\ell_1}\}, \{C_{2,1}, \ldots, C_{2,\ell_2}\}, \ldots, \{C_{\ell,1},\ldots,C_{\ell,\ell_{\ell}}\}$ with the following properties.  
   \begin{enumerate}[label=(\roman*)]
   \item The groups cover all points from~$C$, i.e., $\bigcup_{i\in[\ell], j\in[\ell_i]} C_{i,j}=C$.
   \item The groups are pairwise disjoint, i.e., $\forall i, i', j, j'$ with $i\neq i'$ or $j\neq j'$, $C_{i,j}\cap C_{i',j'}=\emptyset$.
   \item For~$i\in[\ell]$, we call the sets~$C_{i,1},\ldots,C_{i,\ell_i}$ the \emph{sets on layer~$i$}. Two different sets from the same layer are more than $2r$ away, i.e., $\forall i\in [\ell], v\in C_{i, j}, v'\in C_{i, j'}$ with $j\neq j'$, $d(v, v')>2r$.
   \item For~$i\in[\ell]$, the maximum diameter of a group on layer $i$ is at most~$h_i$, i.e., $\max_{j} \max_{v, v'\in C_{i, j}} d(v, v')\le h_i$.
   \end{enumerate} 
\end{definition}

Assume that Algorithm~\ref{alg:guess_general} has been executed for some radius~$r$ and let its output be the center set~$C$ with~$|C|\le k$ and the corresponding sets~$T_c$ for $c\in C$. Furthermore assume that we have computed an $r$-well-separated partition of the center set~$C$ into $\ell$ layers with parameters~$(h_1,\ldots,h_{\ell})$.

Since two centers~$c$ and~$c'$ from different groups of the same layer have a distance of more than~$2r$, their corresponding clusters~$T_c$ and~$T_{c'}$ do not intersect. However, the clusters belonging to two centers in the same group or to groups from  different layers can intersect. The following lemma describes an algorithm that adjusts the clusters layer by layer to make them pairwise disjoint.

\begin{lemma}
   \label{lem:approx_partition}
   Consider an instance $(G=(V,E),M=(V,d),k)$ of the connected $k$-center problem and assume that Algorithm~\ref{alg:guess_general} computes a center set~$C\subseteq V$ with~$|C|\le k$ for some radius~$r$. Furthermore, let an $r$-well-separated partition of~$C$ with $\ell$ layers and parameters~$(h_1,\ldots,h_{\ell})$ be given. Then we can efficiently find  a feasible solution for the connected $k$-center problem with disjoint clusters with radius at most $(2\ell-1)r+\sum_{i=1}^{\ell}h_i$.
\end{lemma}
\begin{proof}
According to Definition~\ref{def:wellseparated} and Algorithm~\ref{alg:guess_general}, we have the following properties:  
\begin{enumerate}[label=(\roman*)]
\item $\bigcup_{i\in[\ell], j\in[\ell_i]} C_{i, j}=C$
\item  $\forall i, i', j, j'$ with $i\neq i'$ or $j\neq j'$: $C_{i, j}\cap C_{i', j'}=\emptyset$
\item $\forall i\in [\ell], c\in C_{i, j}, c'\in C_{i, j'}$ with $j\neq j'$: $d(c, c')>2r$ and $T_c\cap T_{c'}=\emptyset$
\item $ \forall i\in [\ell]$, $j\in[\ell_i]$, $c, c'\in C_{i, j}$: $d(c, c')\le h_i$
\item $\bigcup_{i\in [\ell], j\in[\ell_i]} \bigcup_{c\in C_{i,j}} T_c =V$
\end{enumerate}  

In the first step, we adjust the clusters by merging all non-disjoint clusters whose centers belong to the same group. To be precise, for each group~$C_{i,j}$ we do the following: As long as there are two different centers~$c\in C_{i,j}$ and $c'\in C_{i,j}$ with~$T_c\cap T_{c'}\neq\emptyset$, we remove $c'$ from $C_{i,j}$ and replace~$T_c$ by $T_c\cup T_{c'}$. That is, we merge the two clusters $T_c$ and $T_{c'}$ and define $c$ as its center. Since centers in the same group on layer~$i$ have a distance of at most~$h_i$, after this step the clusters in each group~$C_{i,j}$ are pairwise disjoint and have a radius of at most~$r+h_i$ and a diameter of at most~$2r+h_i$. They are still connected because we only merge connected clusters that have at least one node in common. 

Since clusters in different groups of the same layer are pairwise disjoint anyway, all clusters on the same layer are pairwise disjoint after this step. Hence, in the next step we only need to describe how clusters from different layers can be made disjoint. For this, it will be helpful to view the clusters as trees. To make this more precise, consider a cluster~$T_c$ with center~$c$. We know that the subgraph of~$G$ induced by~$T_c$ is connected. For any cluster~$T_c$ we choose an arbitrary spanning tree in this induced subgraph and consider~$c$ to be the root of this tree. Let $\mathcal{T}_i$ denote the set of all such trees in the $i$-th layer for $i\in [\ell]$. In the following we will use the terms clusters and trees synonymously. By abuse of notation we will use~$T_c$ to denote both the cluster with center~$c$ and the spanning tree with root~$c$, depending on the context.

For every $i\in[\ell]$, all trees in $\mathcal{T}_i$ are node-disjoint. We will now describe how to ensure that  trees on different layers are also node-disjoint. For this, we will go through the layers~$i=1,2,\ldots,\ell$ in this order and replace $\mathcal{T}_i$ by an adjusted set of trees $\mathcal{T}'_i$. We will construct these trees so that at each step~$i\in[\ell]$ all trees from $\cup_{j\in[i]}\mathcal{T}'_j$ are pairwise disjoint. Furthermore, at step~$i$ the radius of any tree from $\cup_{j\in[i]}\mathcal{T}'_j$ will be bounded from above by $(2i-1)r+\sum_{j\in [i]}h_j$. Finally, our construction ensures that in the end, the trees in $\cup_{i\in[\ell]}\mathcal{T}'_i$ cover all nodes in~$V$. Hence, these trees form a feasible solution to the connected $k$-center problem with disjoint clusters with the desired radius.

We set $\mathcal{T}'_1=\mathcal{T}_1$. Then for~$i=1$, the desired properties are satisfied because the trees on layer~1 are pairwise disjoint and have a radius of at most~$r+h_1$. Now assume that the properties are true for some~$i$ and let us discuss how to ensure them also for~$i+1$. We start with $\mathcal{T}_{i+1}'=\emptyset$ and add trees to it one after another. Consider an arbitrary tree $T\in \mathcal{T}_{i+1}=(V',E')$ with center~$c$ and let~$V^*\subseteq V'$ denote the nodes that also occur in some tree $T'\in \mathcal{T}_j'$ for some~$j\in[i]$. Observe that any node from~$V^*$ can be contained in at most one such tree~$T'$ because by the induction hypothesis all trees in~$\cup_{j\in[i]}\mathcal{T}_j'$ are pairwise disjoint. If $V^*$ is empty then the tree~$T$ is disjoint from all trees in~$\cup_{j\in[i+1]}\mathcal{T}_j'$ and does not need to be adjusted. In this case we simply add it to $\mathcal{T}_{i+1}'$.

If $V^*$ contains only a single node~$v$ then we merge the tree~$T$ with the unique tree~$T'$ from $\mathcal{T}_j'$ for some~$j\le i$ that also contains node~$v$, i.e., we replace $T'$ by~$T\cup T'$ in $\mathcal{T}_j'$. Tree~$T'$ has a radius of at most $(2i-1)r+\sum_{j\in [i]}h_j$. Since the diameter of~$T$ is at most~$2r+h_{i+1}$, the radius of the union of~$T$ and~$T'$ with respect to the center of~$T'$ is at most (see Figure~\ref{fig:GeneralGraphsMerge})
\begin{equation}\label{eqn:RadiusMerge}
  (2r+h_{i+1}) +(2i-1)r+\sum_{j\in [i]}h_j = 
  (2(i+1)-1)r+\sum_{j\in [i+1]}h_j.
\end{equation}

\begin{figure}
\centering
\includegraphics[scale=0.7]{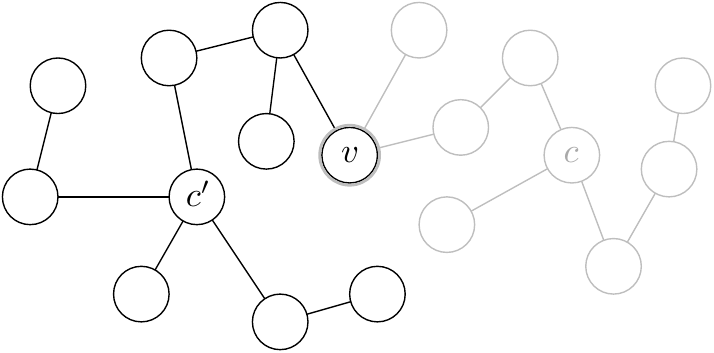}
\caption{This figure shows the tree~$T'$ with center~$c'$ in black and the tree~$T$ with center~$c$ in gray. These trees have node~$v$ in common. When $T$ and~$T'$ are merged into a single tree, the radius of this new tree with respect to~$c'$ is larger than the radius of~$T'$ by at most the diameter of~$T$.}
\label{fig:GeneralGraphsMerge}
\end{figure}

\begin{figure}
\centering
\includegraphics[scale=0.7]{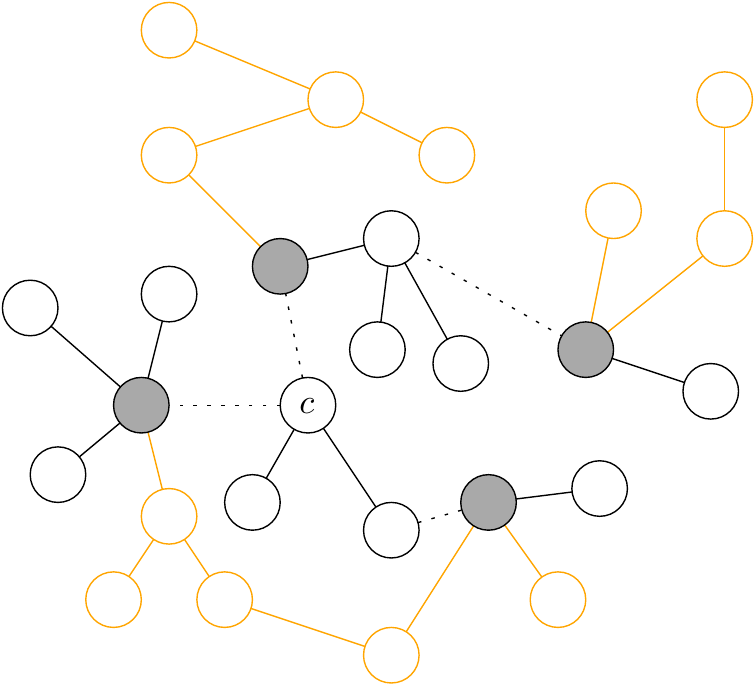}
\caption{This figure shows the tree~$T$ in black. The nodes in~$V^*$ are marked gray and the edges that are removed from~$T$ are shown dotted. The orange trees depict the trees on lower layers that contain the nodes from~$V^*$ and with which the corresponding components are merged.}
\label{fig:GeneralGraphsMerge2}
\end{figure}

Now consider the case that~$V^*$ contains more than one node. In this case we cannot simply merge~$T$ with some tree from~$\cup_{j\in[i]}\mathcal{T}_j'$ because the resulting tree would not be disjoint from the other trees. We also cannot merge all trees that contain nodes from~$V'$ into a single cluster because the radius of the resulting cluster could be too large. Instead we split the tree~$T$ into multiple components and we merge these components separately with different trees from $\cup_{j\in[i]}\mathcal{T}_j'$. For each node~$v\in V^*$ that is not the root~$c$ of~$T$ we consider the path from~$c$ to~$v$ and let~$e$ denote the last edge on this path (i.e., the edge leading to~$v$). We remove edge~$e$ from the tree~$T$ and thereby split the tree~$T$ into two components. Since we do this for every node from~$V^*\setminus\{c\}$, the tree~$T$ will be split into~$|V^*\setminus\{c\}|+1$ pairwise disjoint connected components. Each of these components that does not contain the root~$c$ contains exactly one node from~$V^*$. Hence, for each of these components there is a unique tree from $\cup_{j\in[i]}\mathcal{T}_j'$ from which it is non-disjoint. We merge every component with the tree from which it is non-disjoint (see Figure~\ref{fig:GeneralGraphsMerge2}). In the component that contains the root, only the root might belong to~$V^*$. If this is the case, we merge it with the unique tree from $\cup_{j\in[i]}\mathcal{T}_j'$ from which it is non-disjoint. Otherwise, we add this component to $\mathcal{T}'_{i+1}$. Since~$T$ has a diameter of at most~$2r+h_{i+1}$, the same is true for each of the components. By the induction hypothesis, each tree from $\cup_{j\in[i]}\mathcal{T}_j'$ has a radius of at most $(2i-1)r+\sum_{j\in [i]}h_j$. Hence, as in~\eqref{eqn:RadiusMerge}, the radius of the merged clusters is bounded from above by $(2(i+1)-1)r+\sum_{j\in [i+1]}h_j$.
\end{proof}

\begin{corollary}\label{cor:ApproximationByPartition}
If there exists a polynomial-time algorithm that computes for any metric~$(C,d)$ and any~$r$ an $r$-well-separated partition with $\ell$ layers and parameters $(h_1,\ldots,h_{\ell})$ then there exists an approximation algorithm for the connected $k$-center problem with disjoint clusters that achieves an approximation factor of $4\ell-2+2\sum_{i=1}^{\ell}h_i/r$.
\end{corollary}
\begin{proof}
To obtain the desired approximation factor, we first determine the smallest~$r$ for which Algorithm~\ref{alg:guess_general} returns a center set~$C$ with~$|C|\le k$. Due to Lemma~\ref{lemma:RadiusWithin2Opt}, this radius~$r$ will be at most~$2r^*$, where~$r^*$ denotes the radius of an optimal connected $k$-clustering with non-disjoint clusters. Let~$r^*_D$ denote the radius of an optimal connected $k$-clustering with disjoint clusters. Then~$r^*_D\ge r^* \ge r/2$. According to Lemma~\ref{lem:approx_partition}, the polynomial-time algorithm for computing an $r$-well-separated partition can then be used to compute a connected $k$-clustering with disjoint clusters and radius at most $(2\ell-1)r+\sum_{i\in [\ell]}h_i$. The approximation factor of this $k$-clustering is
\[
   \frac{(2\ell-1)r+\sum_{i\in [\ell]}h_i}{r^*_D} \le 
   \frac{(2\ell-1)r+\sum_{i\in [\ell]}h_i}{r/2}
   = 4\ell-2+2\sum_{i\in [\ell]}\frac{h_i}{r}. \qedhere
\]
\end{proof}

\subsection{Extending the Analysis to Connected \texorpdfstring{$k$}{k}-diameter Clustering}

The same algorithm that we developed in the previous sections for the connected $k$-center problem can also be used for the connected $k$-diameter problem without any modifications. Only the analysis of the approximation factor needs to be adapted slightly.

Lemma~\ref{lemma:RadiusWithin2Opt} is changed as follows.
\begin{lemma}\label{lemma:DiameterWithinOpt}
Let~$r^*$ denote the diameter of an optimal connected $k$-diameter clustering with non-disjoint clusters. For~$r\ge r^*$, Algorithm~\ref{alg:guess_general} computes a center set~$C$ with~$|C|\le k$.
\end{lemma}
Observe that the diameter of the clusters~$T_c$ that are computed by Algorithm~\ref{alg:guess_general} for some~$r$ can be at most~$2r$.

A straightforward adaption of Lemma~\ref{lem:approx_partition} yields the following result.
\begin{lemma}
   \label{lem:approx_partitionDiameter}
   Consider an instance $(G=(V,E),M=(V,d),k)$ of the connected $k$-diameter problem and assume that Algorithm~\ref{alg:guess_general} computes a center set~$C\subseteq V$ with~$|C|\le k$ for some radius~$r$. Furthermore, let an $r$-well-separated partition of~$C$ with $\ell$ layers and parameters~$(h_1,\ldots,h_{\ell})$ be given. Then we can efficiently find a feasible solution for the connected $k$-diameter problem with disjoint clusters with diameter at most $(4\ell-2)r+h_1+2\sum_{i=2}^{\ell}h_i$.
\end{lemma}

To see that this adapted lemma is true, first consider the diameter of the clusters that are created in the first step by merging overlapping clusters from the same group. Since centers in the same group on layer~$i$ have a distance of at most~$h_i$, the diameter of the merged clusters is at most $2r+h_i$. To see this, consider two centers~$c$ and~$c'$ of merged clusters. Then~$c$ and~$c'$ have a distance of at most~$h_i$. Now consider any points~$v\in T_c$ and~$v'\in T_{c'}$. We then have
\[
   d(v,v') \le d(v,c) + d(c,c') + d(c',v)
           \le r+h_i+r = 2r+h_i.
\]
Hence, the lemma is true for~$\ell=1$.

Now only~\eqref{eqn:RadiusMerge} needs to be adapted: Using the notation from the proof of Lemma~\ref{lem:approx_partition}, we know that~$T$ has a diameter of at most~$2r+h_{i+1}$ and~$T'$ has a diameter of at most $(4i-2)r+h_1+2\sum_{j=2}^{i}h_j$ by the induction hypothesis. A tree~$T'$ from $\cup_{j\in[i]}\mathcal{T}'_j$ might get merged with multiple trees from $\mathcal{T}_{j+1}$. Hence, the diameter of the resulting cluster is at most
\[
   (4i-2)r+h_1+2\sum_{j=2}^{i}h_j + 2(2r+h_{i+1})
   = (4(i+1)-2)r+h_1+2\sum_{j=2}^{i+1}h_j.
\]

Overall we obtain the following corollary.
\begin{corollary}\label{cor:ApproximationByPartitionDiameter}
If there exists a polynomial-time algorithm that computes for any metric~$(C,d)$ and any~$r$ an $r$-well-separated partition with $\ell$ layers and parameters $(h_1,\ldots,h_{\ell})$ then there exists an approximation algorithm for the connected $k$-diameter problem with disjoint clusters that achieves an approximation factor of $4\ell-2+h_1/r+2\sum_{i=2}^{\ell}h_i/r$.
\end{corollary}

\subsection{Finding Well-separated Partitions}

In this section, we will describe how to efficiently compute $r$-well-separated partitions. First we will deal with general metrics and then we will discuss the Euclidean metric and metrics with constant doubling dimension.

\subsubsection{Well-separated Partitions in General Metrics}   
According to Corollaries~\ref{cor:ApproximationByPartition} and~\ref{cor:ApproximationByPartitionDiameter}, we only need to find an efficient algorithm for computing an $r$-well-separated partition to obtain an approximation algorithm for the connected $k$-center and $k$-diameter problem.

Algorithm~\ref{alg:partition_generalmetric} computes an $r$-well separated partition layer by layer. For each layer it creates the groups in a greedy fashion: At the beginning of a layer~$i$, the set~$U'$ of all nodes that are not assigned to previous layers is considered. The goal is to assign as many of these nodes to the current layer~$i$ as possible. For this, we start with an arbitrary node~$u\in U'$ that is not assigned to any previous layer and we create a group around~$u$. First the group consists only of~$u$ itself. Then we iteratively augment the group by adding all nodes to the group that have a distance of at most~$2r$ from some node that already belongs to the group. We repeat this augmentation step multiple times one after another. We stop when the number of new nodes that join the group is smaller than twice the number of nodes that have already been added to the group for the first time. Then the group around~$u$ is finished and added to layer~$i$. All nodes in the group are removed from~$U'$. Furthermore, we also remove all nodes that have a distance of at most~$2r$ from this group from~$U'$. These nodes have to be assigned to other layers that are created later to ensure property~(iii) in Definition~\ref{def:wellseparated}. As long as $U'$ is not empty, we repeat the process to create another group on layer~$i$. The pseudocode is shown as Algorithm~\ref{alg:partition_generalmetric}.

\begin{algorithm}  
   \caption{\textsc{PartitionGeneralMetric}$((C,d),r)$}\label{alg:partition_generalmetric}
   \KwIn{metric $(C,d)$, radius $r$}

   $U \gets C$;     \tcp{nodes that still have to be assigned}
   $i \gets 0$; 

   \While{$U \neq \emptyset$}
   {  
      $i\gets i+1$;   \tcp{start a new layer}
      
      $j\gets 0$\;

      $U'\gets U$;   \tcp{nodes that could still be assigned on $i$-th layer}

      \While{$U'\neq \emptyset$}
      {
         $j=j+1$;  \tcp{create a new group in $i$-th layer}

         select a node $u\in U'$, $C_{i, j}\gets \{u\}$ and  $N_{0}(u) \gets \{u\}$\;
         
         $U'\gets U'\setminus \{u\}$; 
  
         $U\gets U\setminus \{u\}$;

         $s=1$\;

         \While{$s\neq 0$ and $U'\neq \emptyset$}
         {
            $N_s(u) \gets \{ x\in U' \mid \exists v\in N_{s-1}(u): d(v,x)\le 2r\} $\;   \tcp{nearby nodes of nodes $C_{i, j}$ in $U'$}  

            \eIf{$|N_s(u)|\ge 2 \cdot |C_{i, j}|$}
            {
               $C_{i,j} \gets C_{i,j} \cup N_s(u)$;  \tcp{add nearby nodes to $C_{i,j}$}
               
               $U'\gets U' \setminus N_s(u)$; 
  
               $U\gets U \setminus N_s(u)$; 

               $s=s+1$\;
            }
            {
               $U'\gets U'\setminus N_s(u)$;
               \tcp{nearby nodes cannot be on $i$-th layer} 

               $s=0$;  \tcp{end group of node $u$}
            }
         }
      }
   }
   \KwOut{$\{C_{1,1}, C_{1,2}, \ldots\}, \{C_{2,1}, C_{2,2}, \ldots\}, \ldots$}
\end{algorithm}

\begin{lemma}
   \label{lem:partition_generalmetric}
   Let~$(C,d)$ be an arbitrary metric with $k:=|C|$ and~$r>0$.
  Let $\ell =1+\lfloor \log_{\frac32}(k)\rfloor$ and $h= 4r\lfloor \log_3{k} \rfloor$. The output of Algorithm~\ref{alg:partition_generalmetric} is an $r$-well-separated partition with at most $\ell$ layers and parameters $(h,\ldots,h)$.
\end{lemma}
\begin{proof}
Let $\{C_{1,1}, \ldots, C_{1,\ell_1}\}, \{C_{2,1}, \ldots, C_{2,\ell_2}\}, \ldots, \{C_{\ell,1},\ldots,C_{\ell,\ell_{\ell}}\}$ denote the output of Algorithm~\ref{alg:partition_generalmetric}. The algorithm ensures that every point from~$C$ is contained in exactly one group~$C_{i,j}$ because when nodes are deleted from~$U$ in Line~18 they have been added to~$C_{i,j}$ in Line~16. Furthermore~$U'$ is always a subset of~$U$ and so no node can be assigned to multiple clusters. Furthermore, Lines~14 and~21 ensure that nodes in different groups of the same layer are more than $2r$ apart. This shows that the properties (i), (ii), and (iii) in Definition~\ref{def:wellseparated} are satisfied.

Next we show property (iv) that the maximum diameter of every group is~$h$.
As long as the number of nearby nodes in $N_s(u)$ is at least twice the number of the previously grouped nodes in~$\cup_{t=1}^{s-1}N_t(u)$, we add these nearby nodes to the current group. As long as this is true we have
\[
   |N_s(u)| \ge 2\cdot \sum_{t=0}^{s-1}|N_t(u)|.
\]
Together with $|N_0(u)|=1$, this implies $|\cup_{t=1}^{s}N_t(u)|\ge 3^{s}$ for every~$s$ by a simple inductive argument. Since this set cannot contain more than~$k=|C|$ nodes, we have $C_{i,j}=\bigcup_{s=1}^{h} N_s(u)$ for some~$h\le \lfloor \log_3{k} \rfloor$. For any~$s\ge 1$, any node in~$N_s(u)$ has a distance of at most~$2r$ from some node in $N_{s-1}(u)$. Since $u$ is the only node in~$N_0(u)$, this implies that any node has a distance of at most~$2rh$ from~$u$. Hence, the diameter of every group is at most $4rh\le 4r\lfloor \log_3{k} \rfloor$. This shows property (iv) in Definition~\ref{def:wellseparated}.

Now it only remains to bound the number of layers of the partition. When a new layer is started, $U'$ is set to $U$, the set of yet unassigned nodes in Line~6. When a group is formed then its current neighbors~$N_s(u)$ get removed from~$U'$ in Line~21. These are exactly the nodes that do not get assigned to the current layer and have to be assigned to other layers afterwards. Since line~21 is only reached if $|N_s(u)|$ is smaller than twice $|C_{i,j}|$, at least one third of the initially unassigned nodes get assigned to groups on the current layer and at most two thirds are postponed to other layers afterwards. This implies that after $\ell$ layers, there are no more than $(\frac{2}{3})^{\ell} \cdot k$ nodes left to be assigned. Hence, the number of layers cannot be more than $1+\lfloor \log_{\frac32}(k)\rfloor$.
\end{proof}

We found out in hindsight that a problem related to computing well-separated partitions has been studied in the context of distributed computing many years ago. Linial and Saks study the problem of decomposing an unweighted graph into multiple blocks with small diameter where the diameter of a block is defined as the largest diameter of one of its connected components~\cite{LinialS93}. They show that every graph with $n$ nodes has a decomposition into $O(\log{n})$ blocks with diameter $O(\log{n})$. The problem of computing a well-separated partition can be reduced to this problem as follows. We define an unweighted graph $G'$ with node set $V$ and connect two nodes if and only if their distance in the metric is at most $2r$. A decomposition of this graph into blocks can then be translated into a well-separated partition. Each block corresponds to one layer of the partition and the different connected components within a block are the groups on that layer. Since they are not connected in $G'$ their distance is more than $2r$ as required. Hence, the algorithm of Linial and Saks, which is actually almost the same algorithm as Algorithm~\ref{alg:partition_generalmetric}, can be used to obtain a well-separated partition with $O(\log{n})$ layers and parameter $h=h_i=O(2r\log{n})$.
  
Based on Corollaries~\ref{cor:ApproximationByPartition} and~\ref{cor:ApproximationByPartitionDiameter}, it is now easy to prove the following theorem.
\begin{theorem}\label{thm:log-approx_disjoint}
   There exists an $O(\log^2{k})$-approximation algorithm for the connected $k$-center problem and for the connected $k$-diameter problem with disjoint clusters.
\end{theorem}
\begin{proof}
According to Lemma~\ref{lem:partition_generalmetric}, one can efficiently compute for any metric an $r$-well-separated partition with at most $\ell$ layers and parameters $(h,\ldots,h)$ for $\ell =1+\lfloor \log_{\frac32}(k)\rfloor=O(\log{k})$ and $h=4r\lfloor \log_3{k} \rfloor=O(r\cdot\log{k})$.

By Corollary~\ref{cor:ApproximationByPartition} this implies that we can efficiently find a solution for the connected $k$-center problem with disjoint clusters with approximation factor
\[
  4\ell-2+2\sum_{i=1}^{\ell}h/r = O(\ell+\ell h/r) = O(\log{k}+\log^2{k})
  =O(\log^2{k}).
\]

By Corollary~\ref{cor:ApproximationByPartitionDiameter}, it also implies that we can efficiently find a solution for the connected $k$-diameter problem with disjoint clusters with approximation factor
\[
  4\ell-2+h_1/r+2\sum_{i=2}^{\ell}h_i/r = O(\ell+\ell h/r) = O(\log{k}+\log^2{k})
  =O(\log^2{k}).\qedhere
\]
\end{proof}

\subsubsection{Well-separated Partitions in Euclidean Metrics} 

In this section, we study how to compute an $r$-well-separated partition if the metric is an $L_p$-metric in the $d$-dimensional space $\mathbb{R}^d$ for some $p\in \{1, 2, \ldots, \infty\}$. 

\begin{lemma}
   \label{lem:partition_Euclidean}
  For any $L_p$-metric in $\mathbb{R}^d$, an $r$-well-separated partition with $d+1$ layers and parameters $(h,\ldots,h)$ with $h=2r d^{1+1/p}$ can be computed in polynomial time.
\end{lemma}
\begin{proof}
First we partition the space $\mathbb{R}^d$ into $d$-dimensional hyperrectangles with side length at most~$2rd$. These hyperrectangles are chosen such that they are pairwise disjoint and that they cover the entire space. We color these hyperrectangles with $d+1$ colors such that no two neighboring hyperrectangles get the same color where also diagonal neighbors are taken into account. Based on this coloring we then create an $r$-well-separated partition as follows: each color corresponds to one layer of the partition and within a layer all nodes that belong to the same hypercube form a group.

For $d=1$, we partition the space $\mathbb{R}$ by pairwise disjoint intervals of length $2r$ that are alternately colored with colors $1$ and $2$. We assume that these intervals are left-closed and right-open. This results in a well-separated partition with $d+1=2$ layers where each group has a diameter of at most $2r$ with respect to any $L_p$-metric. Since the intervals are half-open, nodes from different intervals with the same color have a distance strictly larger than $2r$.

For $d\ge 2$ the construction is inductive. We take the $(d-1)$-dimensional partition of $\mathbb{R}^{d-1}$ and make it a $d$-dimensional object by extending all hyperrectangles in this partition in dimension $d$ by a length of $2r$. Formally, a $(d-1)$-dimensional hyperrectangle $H\subseteq\mathbb{R}^{d-1}$ becomes the $d$-dimensional hyperrectangle $H\times[0,2r)\subseteq\mathbb{R}^{d}$. These hyperrectangles together form a partition of $\mathbb{R}^{d-1}\times[0,2r)$ with $d$ colors. Let $P(1,\ldots,d)$ denote this partition, where the parameters $1,\ldots,d$ denote the different colors that occur in the partition. In the following, we will use the same partition with different colors, i.e., $P(j_1,\ldots,j_d)$ denotes the partition in which color $i$ gets replaced by color $j_i$ for every $i\in[d]$. We extend $P(1,\ldots,d)$ to a partition of the entire space $\mathbb{R}^d$ by stacking copies of it (with exchanged colors) on top of each other according to the following scheme:
\begin{itemize}
  \item We use $P(1,\ldots,d)$ to cover $\mathbb{R}^{d-1}\times[0,2r)$.
  \item We use $P(d+1,2,\ldots,d)$ to cover $\mathbb{R}^{d-1}\times[2r,4r)$.
  \item We use $P(d+1,1,3\ldots,d)$ to cover $\mathbb{R}^{d-1}\times[4r,6r)$.
  \item We use $P(d+1,1,2,4\ldots,d)$ to cover $\mathbb{R}^{d-1}\times[6r,8r)$.
  \item \ldots
\end{itemize}
Formally the pattern can be described as follows: In each step there is a pointer pointing to one of the color parameters. This color is exchanged by the unique color from $\{1,\ldots,d+1\}$ that is currently not present in the parameter vector. Then the pointer is moved one position to the right or to the first position if it is already at the rightmost position. Afterwards neighboring hyperrectangles with the same color are treated as a single hyperrectangle (observe that there cannot be any diagonal neighbors). This scheme can analogously be used to also cover $\mathbb{R}^{d-1}\times\mathbb{R}_{<0}$ or one could assume without loss of generality that all data points lie in the positive orthant in which case only a partition of $\mathbb{R}_{\ge0}^d$ is needed. Figure~\ref{fig:HypercubeColoring} shows an illustration of our construction.

\begin{figure}
\centering
\includegraphics[scale=0.5]{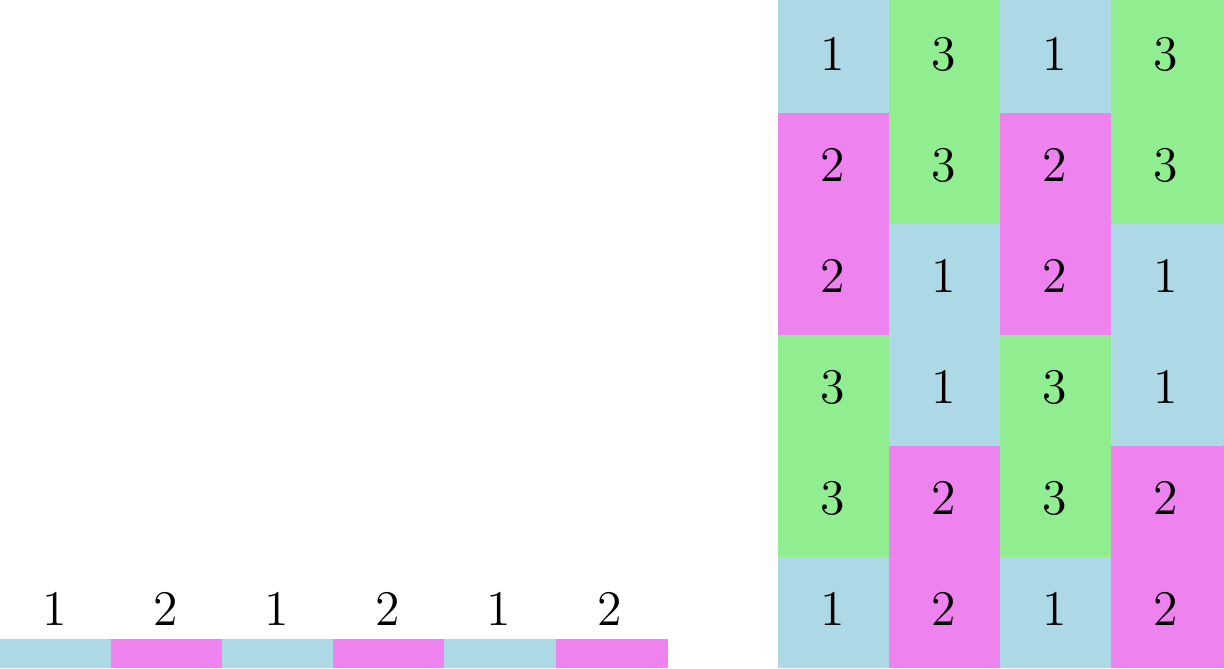}
\caption{On the left the construction for $d=1$ is shown while on the right the construction for $d=2$ is illustrated.}
\label{fig:HypercubeColoring}
\end{figure}

Let us now consider the properties of a well-separated partition. Properties~(i) and~(ii) are satisfied because the hyperrectangles partition the space~$\mathbb{R}^d$. Since neighboring hyperrectangles of the same color are treated as a single hyperrectangle and the intervals are half-open, the distance between two hyperrectangles of the same color is larger than the shortest possible side length of any of the hyperrectangles. Since any side length is at least $2r$, property~(iii) of Definition~\ref{def:wellseparated} is satisfied. It remains to bound the diameter of the hyperrectangles. The initial side length for $d=1$ is $2r$ and in the inductive construction for $d\ge 2$ the side length in dimension $d$ is $2r$ as well. However, the side length in dimension $d$ can grow in the construction because multiple hyperrectangles of the same color might be glued together. Observe that in our construction each color parameter is changed after $d$ steps because the pointer moves cyclically through the parameters. This means that only $d$ of the hyperrectangles can be glued together, resulting in a side length of $2rd$ in dimension $d$. Hence, the diameter of the hyperrectangles is
\[
  \left(\sum_{i=1}^d (2ri)^p\right)^{1/p}
  = 2r\left(\sum_{i=1}^d i^p\right)^{1/p}
  \le 2rd^{1+1/p},
\]
which proves property~(iv) for $h=2r d^{1+1/p}$.
\end{proof}

Based on Corollaries~\ref{cor:ApproximationByPartition} and~\ref{cor:ApproximationByPartitionDiameter}, it is now easy to prove the following theorem.
\begin{theorem}\label{thm:euclidean_approx}
For any $L_p$-metric in~$\mathbb{R}^d$, there exists an $O(d^{2+1/p})$-approximation algorithm for the connected $k$-center problem and for the connected $k$-diameter problem with disjoint clusters.
\end{theorem}
\begin{proof}
According to Lemma~\ref{lem:partition_Euclidean}, we can efficiently compute an $r$-well-separated partition with~$d+1$ layers and parameters $(h,\ldots,h)$ for $h=2r d^{1+1/p}$.

By Corollary~\ref{cor:ApproximationByPartition} this implies that we can efficiently find a solution for the connected $k$-center problem with disjoint clusters with approximation factor
\[
  4\ell-2+2\sum_{i=1}^{d+1}h/r = O(d^{2+1/p}).
\]

By Corollary~\ref{cor:ApproximationByPartitionDiameter}, it also implies that we can efficiently find a solution for the connected $k$-diameter problem with disjoint clusters with approximation factor
\[
  4d-2+h/r+2\sum_{i=2}^{d+1}h/r = O(d^{2+1/p}).\qedhere
\]
\end{proof}

\subsubsection{Well-separated Partitions in Metrics with Small Doubling Dimension} 

In this section, we study how to compute an $r$-well-separated partition if the metric has constant doubling dimension. This generalizes Lemma~\ref{lem:partition_Euclidean} for Euclidean spaces.

\begin{definition}[doubling dimension]
The \emph{doubling constant} of a metric space~$M=(X,d)$ is the smallest number~$k$ such that for all $x \in X$ and $r > 0$, the ball $B_r(x):=\{y\in X\mid d(x,y)\le r\}$ can be covered by at most $k$ balls of radius~$r/2$, i.e.,
\[
  \forall x\in X: \forall r > 0: \exists Y\subseteq X, |Y|\le k:
  B_r(x)\subseteq \bigcup_{y\in Y} B_{r/2}(y).
\]
The \emph{doubling dimension} of~$M$ is defined as $\mathrm{dim}(M)=\lceil\log_2{k}\rceil$.
\end{definition}

\begin{lemma}
\label{lem:partition_Doubling}
For any metric~$M=(X,d)$ with doubling dimension~$\mathrm{dim}(M)$, an $r$-well-separated partition with $2^{3\cdot\mathrm{dim}(M)}$ layers and parameters $(h,\ldots,h)$ with $h=2r$ can be computed in polynomial time.
\end{lemma}
\begin{proof}
First we partition~$X$ greedily into balls of radius~$r$: As long as not all points of~$X$ are covered, we choose arbitrarily an uncovered point~$x$ from~$X$ and put~$x$ into one group together with all uncovered points that have a distance of at most~$r$ from~$x$. This way, we get a partition of~$X$ into groups with radius at most~$r$. 

Next, we try to reduce the number of groups by local improvements. We say that two groups are neighboring if the distance of their centers is at most~$4r$. As long as there is a group that has at least $2^{3\cdot\mathrm{dim}(M)}$ neighbors, we replace this group and its neighbors by $2^{3\cdot\mathrm{dim}(M)}$ groups as follows: Let~$x$ be a center of a group that has at least $2^{3\cdot\mathrm{dim}(M)}$ neighbors, and let the centers of the neighbors be $Y\subseteq X$. Since $x$ has a distance of at most~$4r$ from all centers in~$Y$, we have
\[
   B_r(x) \cup_{y\in Y} B_r(y) \subseteq B_{5r}(x).
\]
By definition of the doubling dimension, the ball~$B_{5r}(x)$ can be covered by $2^{\mathrm{dim}(M)}$ balls of radius $5r/2$, each of these can be covered by $2^{\mathrm{dim}(M)}$ balls of radius $5r/4<2r$, and each of these can be covered by $2^{\mathrm{dim}(M)}$ balls of radius $5r/8<r$. Hence, the points in $B_r(x) \cup_{y\in Y} B_r(y)$ can be covered by $2^{3\cdot\mathrm{dim}(M)}$ balls of radius~$r$. In our partition, we replace the groups around~$x$ and around~$y\in Y$ by the groups induced by these balls. Since this reduces the number of groups by at least one, after a linear number of these local improvements, no local improvement is possible anymore, i.e., every group has less than $2^{3\cdot\mathrm{dim}(M)}$ neighbors.

We have obtained a partition of~$X$ into groups, where each group has a radius of at most~$r$. Furthermore, each group has a center and two groups are neighbors if their centers have a distance of at most~$4r$. Furthermore, every group has less than $2^{3\cdot\mathrm{dim}(M)}$ neighbors. The groups will form the groups in the $r$-well-separated partition. Since points from groups that are not neighbored have a distance of more than~$2r$, two groups that are not neighbored can be on the same layer of the partition without contradicting property (iii) from Definition~\ref{def:wellseparated}. The diameter of each group is at most~$h=2r$. It remains to distribute the groups to the different layers of the partition. For this we find a coloring of the groups such that neighboring groups get different colors. The neighborhood defines implicitly a graph with the groups as vertices with degree at most $2^{3\cdot\mathrm{dim}(M)}-1$. Any such graph can be colored with $2^{3\cdot\mathrm{dim}(M)}$ colors by a greedy algorithm. Now we assign the groups according to the colors to different layers, resulting in an $r$-well-separated partition with at most $2^{3\cdot\mathrm{dim}(M)}$ layers.
\end{proof}

Based on Corollaries~\ref{cor:ApproximationByPartition} and~\ref{cor:ApproximationByPartitionDiameter}, it is now easy to prove the following theorem.
\begin{theorem}\label{thm:doubling_dimension}
For any metric~$M=(X,d)$ with doubling dimension~$\mathrm{dim}(M)$, there exists an $O(2^{3\cdot\mathrm{dim}(M)})$-approximation algorithm for the connected $k$-center problem and for the connected $k$-diameter problem with disjoint clusters.
\end{theorem}
\begin{proof}
According to Lemma~\ref{lem:partition_Euclidean}, we can efficiently compute an $r$-well-separated partition with~$2^{3\cdot\mathrm{dim}(M)}$ layers and parameters $(h,\ldots,h)$ for $h=2r$.

By Corollary~\ref{cor:ApproximationByPartition} this implies that we can efficiently find a solution for the connected $k$-center problem with disjoint clusters with approximation factor
\[
  4\ell-2+2\sum_{i=1}^{\ell}h/r = O(2^{3\cdot\mathrm{dim}(M)}).
\]

By Corollary~\ref{cor:ApproximationByPartitionDiameter}, it also implies that we can efficiently find a solution for the connected $k$-diameter problem with disjoint clusters with approximation factor
\[
  4\ell-2+h_1/r+2\sum_{i=2}^{\ell}h_i/r = O(2^{3\cdot\mathrm{dim}(M)}).\qedhere
\]
\end{proof}

\subsubsection{Well-separated Partitions for Small Number of Clusters} 

In this section, we study how to compute well-separated partitions if the number of clusters is small, particularly when $k=2$. A first observation to improve the approximation factor is that for constant $k$, one can test in polynomial time all possible center sets $C\subseteq V$ with $|C|\le k$ in algorithm \texttt{GreedyClustering} (Algorithm~\ref{alg:guess_general}). Here, we assume that the algorithm has been modified as follows: instead of choosing arbitrary uncovered nodes as centers, it chooses exactly the nodes from~$C$ as centers and calls \texttt{ComputeCluster} once for every such node to form a cluster of radius~$r$. 

Let~$r^*$ denote the radius of an optimal connected $k$-center clustering with non-disjoint clusters. If~$C$ coincides with the center set of such an optimal clustering, then already for $r\ge r^*$ \texttt{GreedyClustering} computes a $k$-clustering covering all points (in Lemma~\ref{lemma:RadiusWithin2Opt} we had to assume $r\ge 2r^*$ because the centers were arbitrarily chosen). Similarly, if we consider the assignment problem for a given center set~$C$ then \texttt{GreedyClustering} computes a $k$-clustering covering all points if $r$ is at least the optimal radius for the given center set~$C$.

A straightforward application of this observation yields the following result, which improves Corollary~\ref{cor:ApproximationByPartition} by a factor of~2.
\begin{corollary}\label{cor:ApproximationByPartitionSmallk}
   If there exists a polynomial-time algorithm that computes for any metric~$(C,d)$ and any~$r$ an $r$-well-separated partition with $\ell$ layers and parameters $(h_1,\ldots,h_{\ell})$ then there exists an approximation algorithm for the connected $k$-center problem with disjoint clusters and given centers that achieves an approximation factor of $2\ell-1+\sum_{i=1}^{\ell}h_i/r$. If the centers are not given, the same factor can be achieved in polynomial time if $k$ is constant.
\end{corollary}

For $k=2$ one can easily compute an $r$-well-separated partition with one layer and parameter~$h_1=2r$: either the two centers have a distance of more than~$2r$, in which case they form two different groups on layer~1, or their distance is at most~$2r$ in which case they are assigned to the same group on layer~1. This implies that we can efficiently find a solution for the connected $k$-center problem with disjoint clusters with approximation factor
\[
   2\ell-1+\sum_{i=1}^{\ell}h_i/r = 3.
\]
For the connected $k$-diameter problem with disjoint clusters for $k=2$, we can directly use Corollary~\ref{cor:ApproximationByPartitionDiameter} to obtain an approximation factor of
\[
  4\ell-2+h_1/r+2\sum_{i=2}^{\ell}h_i/r = 4. 
\]

\begin{corollary}
   \label{cor:3and4approx-2givencenters}
For $k=2$, there exists a $3$-approximation algorithm for the connected $k$-center problem with disjoint clusters with and without given centers. For $k=2$, there also exists a $4$-approximation algorithm for the connected $k$-diameter problem with disjoint clusters.
\end{corollary}

We can further improve the theorem for the $k$-center problem for $k=2$ when the centers are not given. We test all possible choices of the center set~$C$. We determine the smallest~$r$ for which there exists a choice of the center set for which \texttt{GreedyClustering} outputs at most two clusters. This gives the optimal solution for the $k$-center problem with non-disjoint clusters. If the two clusters are disjoint then this is also an optimal solution for the $k$-center problem with disjoint clusters. If not, we merge the two clusters into a single cluster and pick a point in the intersection as the new center. With respect to this point, the new cluster has a radius of at most~$2r$. Hence, we obtain a $2$-approximation of the optimal disjoint solution.

\begin{corollary}
   \label{cor:2approx-2freecenters}
For $k=2$, there exists a $2$-approximation algorithm for the connected $k$-center problem with disjoint clusters.
\end{corollary}

\subsection{Lower bound for transforming a non-disjoint clustering into a disjoint clustering} \label{sec:WorstCaseInstance}

In the following, we will describe an instance for the connected $k$-center problem and the connected $k$-diameter problem with disjoint clusters with the following property: There exists a set~$C$ of $k$ centers that can be produced by the algorithm \texttt{GreedyClustering} (Algorithm~\ref{alg:guess_general}) such that even the optimal disjoint solution with centers~$C$ is worse than the optimal disjoint solution for arbitrary centers by a factor of~$\Omega(\log\log{k})$. Hence, to prove a constant factor approximation one cannot rely on the centers chosen by \texttt{GreedyClustering}.

We split the proof into the following two claims: First, we show that for the instance we construct there exists a set of $k$ centers that can be obtained by Algorithm~\ref{alg:guess_general} with radius $r$ such that the optimal disjoint clusters with respect to these centers have a radius of $\Omega(\log \log k) \cdot r$. Second, we show that there exists another set of $k$ centers that allows for a disjoint solution with radius $2r$.

In the following, we use the notation $[n]:=\{1,\ldots,n\}$. Before we introduce the instance, we first define a formula $S(t)$: $S(1) = 0$, $S(2) = 1$ and $S(t+1) = S(t) \cdot (S(t) + 1)$ for $t \ge 2$, and a metric space $(\mathcal{V}_m, d)$ with set $\mathcal{V}_m=([S(m) + 1] \cup \{\perp\}) \times ([S(m - 1) + 1] \cup \{\perp\}) \times \ldots \times ([S(1) + 1] \cup \{\perp\})$ and metric $d$: each point $a\in M$ is represented by a vector $(a_1, a_2, \ldots, a_{m})$; the distance is defined as $d(a, b)=\sum_{i\in [m]} d'(a_i, b_i)$ for any $a, b\in \mathcal{V}_m$ where 
\begin{equation*}
  d'(a_i,b_i) = \begin{cases}
        0 & \text{if $a_i = b_i$,}\\
        1 & \text{if $a_i \neq b_i$ and $(a_i = \perp \text{or}~b_i = \perp)$,}\\
        2 & \text{else.} 
      \end{cases}
  \end{equation*} 
The distance $d$ is indeed a metric on $\mathcal{V}_m$ because $d'$ is a metric for each coordinate~$i$ (the triangle inequality is true because different $a_i$ and $b_i$ have distance $1$ or $2$) and $d$ is the sum of~$d'$ for the different coordinates. 
Let $X(a_i)$ be an indicator variable: 
\begin{equation*}
  X(a_i) = \begin{cases}
        1 & \text{if $a_i =  \perp$,}\\
        0 & \text{else.}  
      \end{cases}
  \end{equation*}  
Let $I(m)$ denote the $m$-th instance with a connectivity graph $G=(V, E)$, metric $(V, d)$, center set $C$ as follows: 
\begin{enumerate}
\item The point set $V$ includes all points such that $\perp$ occurs at most once, i.e.,\\ $V=\{a\mid \forall a\in \mathcal{V}_m, \sum_{i\in [m]} X(a_i) \le 1\}$.\\ Note that $|V|=\prod_{i\in [m]} (S(i)+1) + \sum_{i\in [m]} \prod_{j\in [m]\setminus\{i\}} (S(j)+1)$. 
\item The center set $C$ consists of all points that do not contain $\perp$, i.e.,\\ $C=\{a\mid \forall a\in \mathcal{V}_m, \sum_{i\in [m]} X(a_i) =0\}$.\\ Note that $|C|=\prod_{i\in [m]} (S(i)+1)=S(m+1)$.  
\item For each point $a\in V$ with $a_i\neq \perp$ for all $i\in [m]$, there exists an edge $(a, b)\in E$ where $b=(\perp, a_2, a_3, \ldots, a_{m})$; for any pair of points $a, b\in V$, there is an edge $(a, b)\in E$ if there exists an index $j\in [m-1]$ with $a_j=b_{j+1}=\perp$ and $a_i=b_i$ for all $i\notin \{j, j+1\}$ (see Figure \ref{fig:worst_case_3}).
\end{enumerate}
It is easy to see that Algorithm~\ref{alg:guess_general} for $r=1$ could compute exactly the center set~$C$ for the above instance $I(m)$. The cluster with center $c\in C$ contains all points in $V$ such that one coordinate of $c$ is replaced by $\perp$, i.e., $\{a\in V\mid \exists i\in[m]: \forall j\in [m]\setminus \{i\}: a_i=\perp, a_j=c_j\}$. 

\begin{figure}
\centering
\begin{minipage}{.45\textwidth}
  \centering
  \includegraphics[width=.8\linewidth]{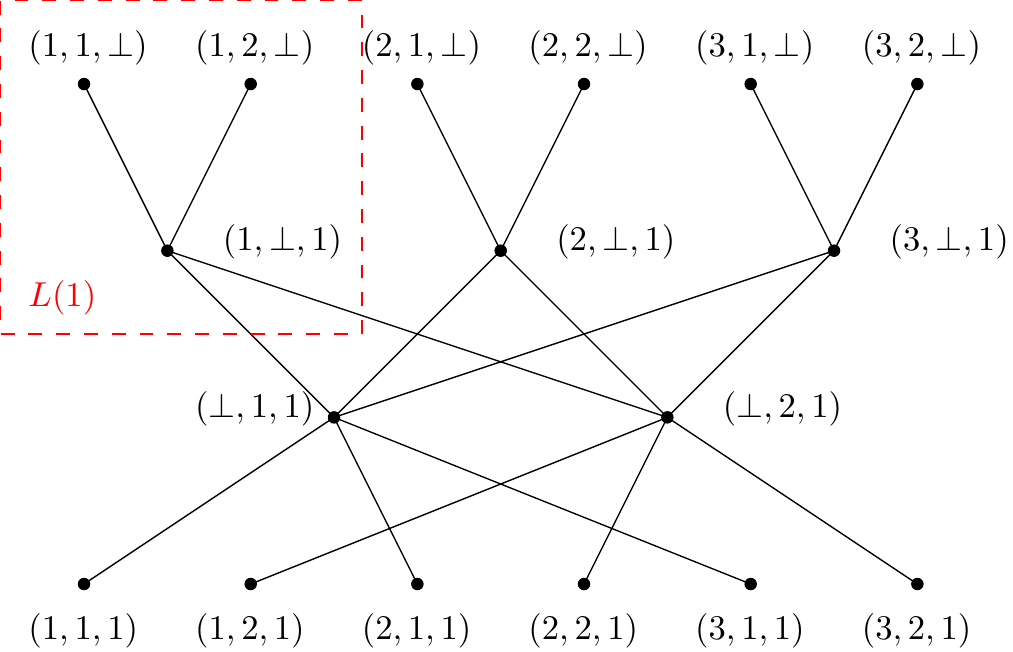}
  \captionof{figure}{The Instance $I(3)$}
  \label{fig:worst_case_3}
\end{minipage}%
\end{figure}
 
Now we are ready to prove the first claim about the gap between the radius of the non-disjoint clusters and the optimal disjoint clusters with the same centers.
\begin{lemma}
  \label{lem:general_worstcase_NtoDisjoint}
  The radius of an optimal solution with disjoint clusters for instance $I_m$ with given centers~$C$ is $2m-1$.    
\end{lemma}
  
\begin{proof}
To prove the statement we will define the following sets of nodes: For an $m' < m$ and $(\alpha_1,\ldots,\alpha_{m'}) \in [S(m) + 1] \times [S(m - 1) + 1]\times \ldots \times [S(m - m' + 1) + 1]$ we define $L(\alpha_1,\ldots,\alpha_{m'}) : = \left\{v \in V \setminus C \mid \forall i \leq m': v_i = \alpha_i\right\}$. For $m'=0$ we write $L(\alpha_1,\ldots,\alpha_{m'})=L()=V\setminus C$.

Let us consider an arbitrary disjoint assignment of the nodes to the center set $C$. We will prove the following statement: For all $m' \in \{0,\ldots,m -1\}$ there exist $\alpha_1,\ldots,\alpha_{m'}$ such that every node in $L(\alpha_1,\ldots,\alpha_{m'})$ gets assigned to a center $c$ with $c_i \neq \alpha_i$ for all $ i \leq m'$. Note that this statement directly proves the lemma, because for $m' = m-1$ we get that there exist a node in $V$ with $\perp$ in the last coordinate that gets assigned to a center $c$ with $c_i \neq v_i$ for all $i \leq m -1$. Since $c$ does not contain a $\perp$ in any coordinate, we may conclude that the distance between $v$ and $c$ is exactly $2(m-1) + 1 = 2m -1$. Thus also the radius of the assignment is lower bounded by $2m -1$.

To prove the statement we will use induction. For $m' = 0$, we have that $L() = V \setminus C$. Furthermore it holds trivially for all $v \in L()$ that they get assigned to a center $c$ with $v_i \neq c_i$ for all $i \leq 0$ since there are no coordinates with an index smaller $1$.

Let us now assume that the statement holds for an $m' < m -1$. Let $\overline{\alpha} = \alpha_1,\ldots,\alpha_{m'}$ be the respective tuple. Note that all edges between nodes in $L(\overline{\alpha})$ and nodes in $V \setminus L(\overline{\alpha})$ are incident to nodes $v \in L(\overline{\alpha})$ with $v_{m' +1} = \perp$. We will call these nodes in the following \emph{decision nodes}. Obviously the nodes $L(\overline{\alpha})$ can only be assigned to the same centers as the decision nodes. There are exactly $\prod_{i= 1}^{m -m' -1} (S(i) +1) = S(m - m')$ decision nodes. Thus we may conclude by the pigeonhole principle that there exists an $\alpha_{m'+1} \in [S(m - m') +1]$ such that no decision node gets assigned to a center $c$ with $c_{m' +1} = \alpha_{m' +1}$. By combining this with the fact that $L(\overline{\alpha},\alpha_{m'+1}) \subseteq L(\overline{\alpha})$ we obtain that every node in $L(\alpha_1,\ldots,\alpha_{m'},\alpha_{m' +1})$ gets assigned to a center $c$ with $\alpha_i \neq c_i$ for $i \leq m' +1$. Thus the lemma is proven by induction.
\end{proof}

  To show the gap $\Omega(\log \log k)$ between the radius of the non-disjoint clusters and the optimal disjoint clusters with given centers~$C$, it remains to bound the number of centers $|C|=S(m+1)$ in instance $I(m)$. According to the definition of $S(m)$, we can show that the inequality $S(m) \le 2^{2^{m-1}}- 2^{2^{m-2}} -1$ for $m\ge 2$ holds by induction: For $m=2$, $S(2)=1=2^{2^1}-2^{2^0}-1$. Suppose the inequality holds for some $m\ge 2$ then by definition we have:
  \begin{align*}
    S(m+1) &= S(m)\cdot (S(m)+1)\\
    &\le ( 2^{2^{m-1}}- 2^{2^{m-2}} -1) \cdot ( 2^{2^{m-1}}- 2^{2^{m-2}})\\
    &= 2^{2^{m}} - 2 \cdot 2^{2^{m-1}} \cdot 2^{2^{m-2}} + 2^{2^{m-2}}   \\
    &= 2^{2^{m}} - (2 \cdot 2^{2^{m-1}} - 1) \cdot 2^{2^{m-2}}   \\
    &\le  2^{2^{m}}- 2^{2^{m-1}} \cdot 2^{2^{m-2}}\\
    &< 2^{2^{m}}.
    \end{align*}
  Combining Lemma~\ref{lem:general_worstcase_NtoDisjoint} and $|C|=S(m+1)< 2^{2^{m}}$, we know that the gap between the radius of the non-disjoint clusters and the optimal disjoint clusters for center set~$C$ in instance $I(m)$ is at least $2m-1 \in \Omega(\log \log |C|) $. 

  Next, we prove the second claim about the radius of an optimal solution with  disjoint clusters without given centers.
  \begin{lemma}
  \label{lem:general_worstcase_Disjoint}
  For the instance $I(m)$, there exists a solution with $|C|$ disjoint clusters with radius $2$ that cover all points in $V$.     
  \end{lemma}
  \begin{proof}
  Obviously for $m = 1$ the statement is fulfilled (by the original center set $C$) and we only need to look at $m \geq 2$. For $I(m)$, consider the center set $C'=\{a\in V\mid \exists i\in [m-1], a_i=\perp\}$, i.e., $C'$ contains all points that have a $\perp$ in some coordinate except the last one. Recall the definition of the connectivity graph: We know that each point without a $\perp$ coordinate is directly connected to a point with $\perp$ sign in the first coordinate, and the distance between them is $1$ because all coordinates are the same except for the first coordinate. If $m \geq 2$ each point with $\perp$ in the $m$-th coordinate is directly connected to a point with $\perp$ in the $(m-1)$-th coordinate, and the distance between them is $2$ because all other coordinates are the same. Thus, when we choose centers $C'$, all remaining nodes are directly connected to a center either with distance $1$ or $2$.  

To bound the number of nodes in $C'$, we may first note that by the definition of $\mathcal{V}_m$ there are $S(m + 1 - i) + 1$ different choices for the $i$-th coordinate that are unequal $\perp$. Thus we can bound the numbers of points with $\perp$ in the $i$-th coordinate by
\begin{align*}
\prod_{j \in [m]\setminus \{i\}} (S(m+1 - j) + 1) &= \prod_{j \in [m]\setminus \{m + 1 - i\}} (S(j) + 1)\\
&= \frac{S(m+1)}{S(m +1 -i) +1}.
\end{align*}
Using this formula we can bound the size of $C'$ as follows:
\begin{align*}
|C'| &= \sum_{i = 1}^{m-1} \frac{S(m+1)}{S(m +1 -i) +1}\\
&= S(m +1) \sum_{i = 2}^m \frac{1}{S(i) + 1}.
\end{align*}

Thus, to complete the proof, we only need to prove $\sum_{i=2}^{m} \frac{1}{S(i)+1} \le 1$ for $m\ge 2$. To do this, we show that $\sum_{i=2}^{m} \frac{1}{S(i)+1}=1- \frac{1}{S(m + 1)}$ holds by induction. For $m=2$, $\sum_{i=2}^{m} \frac{1}{S(i)+1} = \frac{1}{S(2)+1} =1/2 =1- \frac{1}{S(3)}$. Suppose the inequality holds for some $m\ge 2$, then by definition we have:
  \begin{align*}
    \sum_{i=2}^{m+1} \frac{1}{S(i)+1} &=  \sum_{i=2}^{m} \frac{1}{S(i)+1}   +  \frac{1}{S(m+1)+1} \\
    &= 1- \frac{1}{S(m+1)} +   \frac{1}{S(m+1)+1}  \\
    &=  1- \frac{1}{S(m+2)}.
    \end{align*} 
  Here, the last equation holds because $\frac{1}{S(m+1)} - \frac{1}{S(m+1)+1} = \frac{1}{S(m+1)(S(m+1)+1)} = \frac{1}{S(m+2)}$.  Thus the equation holds by induction on $m$ and the lemma is proven.  
  \end{proof}
  
\subsection{The Gap between the optimum non-disjoint and disjoint clustering} \label{sec:WorstCaseInstancePartII}

In the last section we described an instance in which the best centers for the non-disjoint clustering turned out to be poor centers for the disjoint clustering which means that any disjoint clustering algorithm aiming for a constant approximation ratio cannot simply rely on the centers produced by a non-disjoint clustering routine. In this instance there exists a different solution with a radius differing from the non-disjoint radius only by a constant (which was necessary to show the lower bound of $\log\log(k)$ for the performance). However, in general this is not the case and one can show that there can be an $\Omega(\log\log(k))$ gap between the radius of the optimum solutions for disjoint and non-disjoint clustering. To show this we will modify the instance $I(m)$ to obtain an instance $I'(m)$ with the respective property. This time instead of the special vector entry $\perp$ we will introduce $k + 1 = S(m +1) + 1$ different special entries $\perp_1,\ldots,\perp_{k+1}$. The entire set of these entries will be denoted as $T_m$. Now we redefine our metric space $(\mathcal{V}_m, d)$ with set $\mathcal{V}_m=([S(m) + 1] \cup T_m) \times ([S(m - 1) + 1] \cup T_m) \times \ldots \times ([S(1) + 1] \cup T_m)$ and metric $d$: each point $a\in M$ is represented by a vector $(a_1, a_2, \ldots , a_{m})$; the distance is defined as $d(a, b)=\sum_{i\in [m]} d'(a_i, b_i)$ for any $a, b\in \mathcal{V}_m$ where 
\begin{equation*}
  d'(a_i,b_i) = \begin{cases}
        0 & \text{if $a_i = b_i$,}\\
        1 & \text{if $a_i \neq b_i$ and $((a_i \in T_m \land b_i \not\in T_m)~\text{or}~(a_i \not\in T_m \land b_i \in T_m))$,}\\
        2 & \text{else.} 
      \end{cases}
\end{equation*} 

Again we may define the indicator variable $X(a_j)$ which is $1$ iff $a_j \in T_m$ and $0$ otherwise. Using this we may define the graph $G = (V,E)$ of the instance $I'(m)$:
\begin{itemize}
\item The point set $V$ includes all points such that there exist at most one coordinate containing an entry from $T_m$, i.e.,\\ $V=\{a\mid \forall a\in \mathcal{V}_m, \sum_{i\in [m]} X(a_i) \le 1\}$.
\item For each point $a\in V$ with $a_i \notin T_m$ for all $i\in [m]$ and every $\perp_j \in T_m$, there exists an edge $(a, b)\in E$ where $b=(\perp_j, a_2, a_3, \ldots, a_{m})$; for any pair of points $a, b\in V$, there is an edge $(a, b)\in E$ if there exists an $\perp_j \in T_m$ and an index $h\in [m-1]$ with $a_h=b_{h+1}=\perp_j$ and $a_i=b_i$ for all $i\notin \{h, h+1\}$.
\end{itemize}

To get a better understanding of the graph $G$, we subdivide the node set $V$ into multiple disjoint subsets. For each $\perp_j \in T_m$ we define $V_j = \{v \in V\mid \exists j \leq m: v_j = \perp_j\}$. The $V_i$ are pairwise disjoint because the points in $V$ contain at most one coordinate with an entry in $T_m$. Additionally we define a point set $A=\{a\mid \forall a\in \mathcal{V}_m, \sum_{i\in [m]} X(a_i) =0\}$. Note that $A$ is exactly the same as the center set $C$ of $I(m)$.

Generally the only difference between $I(m)$ and $I'(m)$ is that in $I'(m)$ every point outside of $C$ has been copied $k + 1$ times. If we only consider the node set $V_j \cup A$ for an arbitrary $\perp_j \in T_m$ and the edges between nodes in this set, we end up with exactly the same graph as in $I(m)$. Additionally for every $V_h$, $V_j$ with $h \neq j$ there exist no direct edges between the two sets and the only connections between the node sets pass through the nodes in $A$. The main idea is now that for the non-disjoint case we can simply use $A$ as a center set and get the same radius as before. But for the disjoint case it is not beneficial anymore to move the centers towards the later layers to obtain a constant radius because there exist more pairwise disjoint sets $V_i$ than centers. This means at least one of these sets will not contain a center and will have to be assigned via the points in $A$ which leads to a non-constant radius. This can be formalized by the following two lemmas.

\begin{lemma}
There exists a non-disjoint $k$-center solution in $I'(m)$ with radius $1$.
\end{lemma}

\begin{proof}
Let us simply choose the center set $A$ (which has size $k$). To each $c \in A$ we assign all nodes where one of the coordinates has been replaced by an $\perp_j$, i.e., the set $\{a\in V\mid \exists i\in[m], j \in [k+1]: \forall h\in [m]\setminus \{i\}: a_i=\perp_j, a_h=c_h\}$. It is easy to verify that the resulting clusters are connected and that every point in $V$ is covered with radius~$1$.
\end{proof}

\begin{lemma}
The minimum radius of a disjoint connected $k$-center solution in $I'(m)$ is at least $2m -2$.
\end{lemma}

\begin{proof}
Let us consider an arbitrary center set $C$. Since there are $k+1$ different sets $V_1,\ldots,V_{k+1}$ and the sets are pairwise disjoint, we may apply the pigeonhole principle to conclude that there exists an $i \in [k+1]$ such that $C \cap V_i = \emptyset$. Note that the only edges from nodes in $V \setminus V_i$ to nodes in $V_i$ are incident to $A$. Thus the nodes in $V_i$ can only be assigned to their respective centers via the nodes in $A$ which were exactly the centers in the instance $I(m)$. Similarly as in the proof of Lemma \ref{lem:general_worstcase_NtoDisjoint} we will define for any $m' < m$ and $(\alpha_1,\ldots,\alpha_{m'}) \in [S(m) + 1] \times [S(m - 1) + 1]\times \ldots \times [S(m - m' + 1) + 1]$ the set $L(\alpha_1,\ldots,\alpha_{m'}) : = \left\{v \in V _i \mid \forall i \leq m': v_i = \alpha_i\right\}$ and for $m' = 0$ we will define $L(\alpha_1,\ldots,\alpha_{m'})=L()= V_i$. Let us now consider an arbitrary assignment. Again we will prove that for all $m' \in \{0,\ldots,m -1\}$ there exist $\alpha_1,\ldots,\alpha_{m'}$ such that every node in $L(\alpha_1,\ldots,\alpha_{m'})$ gets assigned to a center $c$ with $c_j \neq \alpha_j$ for all $ j \leq m'$. By combining this statement for $m' = m - 1$ with the fact that $c$ contains at most one coordinate in $T_m$ we get that the radius of the assignment is at least $2(m-2)+2=2m -2$.

To prove the statement we will use induction. For $m' = 0$, we have that $L() = V_i$. Furthermore it holds trivially for all $v \in L()$ that they get assigned to a center $c$ with $v_i \neq c_i$ for all $i \leq 0$ since there are no coordinates with an index smaller $1$.

Let us now assume that the statement holds for an $m' < m -1$. Let $\overline{\alpha} = \alpha_1,\ldots,\alpha_{m'}$ be the respective tuple. Note that all edges between nodes in $L(\overline{\alpha})$ and nodes in $V \setminus L(\overline{\alpha})$ are incident to nodes $v \in L(\overline{\alpha})$ with $v_{m' +1} = \perp$. We will call these nodes in the following \emph{decision nodes}. Since $V_i$ and thus also $L(\overline{\alpha})$ contain no centers, the nodes in $L(\overline{\alpha})$ can only be assigned to the same centers as the decision nodes. There are exactly $\prod_{i= 1}^{m -m' -1} (S(i) +1) = S(m - m')$ decision nodes. Thus we may conclude by the pigeonhole principle that there exists an $\alpha_{m'+1} \in [S(m - m') +1]$ such that no decision node gets assigned to a center $c$ with $c_{m' +1} = \alpha_{m' +1}$. By combining this with the fact that $L(\overline{\alpha},\alpha_{m'+1}) \subseteq L(\overline{\alpha})$ we obtain that every node in $L(\alpha_1,\ldots,\alpha_{m'},\alpha_{m' +1})$ gets assigned to a center $c$ with $\alpha_i \neq c_i$ for $i \leq m' +1$. Thus the lemma is proven by induction.
\end{proof}

It was already shown in the previous section that $k = S(m+1) \leq 2^{2^m}$. Thus $2m -2$ lies in $\Omega\left(\log\log(k)\right)$ and we obtain the following result.

\begin{theorem}
There exist instances for the connected $k$-center problem such that the ratio between the radius of the optimum disjoint connected $k$-center solution and the optimum non-disjoint connected $k$-center solution lies in $\Omega\left(\log\log(k)\right)$. 
\end{theorem}

\section{Approximation Hardness of Connected \texorpdfstring{$k$}{k}-clustering Problems}
\label{sec:LowerBounds}

In this section we will prove several lower bounds for the approximability of connected $k$-center and $k$-diameter clustering. At first we will focus on the assignment version of the connected $k$-center problem with disjoint clusters. In this version, the centers are already given beforehand and all that is left is to assign the nodes to the respective clusters. Intuitively one would think that this is easier than the problem in which we also need to find suitable centers and indeed for the regular $k$-center problem it is trivially possible to optimally solve this problem. However, we will show via a reduction from the 3-SAT problem that it is NP-hard to approximate the connected $k$-center problem with disjoint clusters with any approximation factor smaller than $3$, even for $k=2$ and given centers. The same reduction also shows that it is NP-hard to approximate the $2$-center problem with disjoint clusters (without given centers) with an approximation factor better than~$2$. This matches the approximation factor from the algorithm in Corollary~\ref{cor:2approx-2freecenters}, showing altogether that for $k=2$ the assignment problem is harder than the problem without given centers. After this, we will modify the reduction to show that it is also NP-hard to approximate the connected $k$-center problem with disjoint clusters (without given centers) with an approximation factor smaller than $3$ if $k \geq 4$. 

Then we drop the condition that the clusters need to be disjoint and show that in this case both the connected $k$-center and $k$-diameter problem are NP-hard to approximate with an approximation factor less than $2$ even if the respective connectivity graph is a star. Interestingly for the disjoint case, we have proven in Section \ref{chap:treeopt} that the connected $k$-center problem can be solved optimally if the connectivity graph is a tree and is thus easier than the non-disjoint case. For the connected $k$-diameter problem, however, we will prove afterwards that also the disjoint version is NP-hard to approximate with an approximation factor better than $2$ if the connectivity graph is a star. Thus the $2$-approximation for the connected $k$-diameter problem on trees is indeed tight.

\subsection{Hardness of the Assignment Problem for Connected \texorpdfstring{$k$}{k}-center}\label{lowerboundskcenter}

\label{lower-bound-assignment-problem}

\begin{theorem}\label{lem:worst_assign}
It is NP-hard to approximate the assignment problem for connected $k$-center clustering with disjoint clusters with an approximation factor smaller than $3$ even if $k = 2$.
\end{theorem}

\begin{proof}
  We reduce the 3-satisfiability problem (3-SAT) to the assignment problem for the connected $2$-center problem with disjoint clusters. For this consider an instance of 3-SAT, i.e., a formula in conjunctive normal form in which every clause consists of at most three literals. Let $X = \{x_1,...,x_n\}$ denote the set of variables and let $C=\{c_1, c_2, ..., c_m\}$ denote the set of clauses, where each clause consists of at most three literals from $\{x_i, \bar{x}_i\}_{i\in [n]}$. The question is if one can assign values to the variables in $X$ such that all clauses in $C$ are satisfied. This problem is NP-complete~\cite{karp1972reducibility}.

  For the given instance of 3-SAT, we create an instance of the assignment problem for the connected $2$-center problem with disjoint clusters as follows. The instance is defined on the point set
  $$V=\{T, F\}\cup \{x_i, \bar{x}_i, a_i\mid {i\in [n]}\}\cup \{b_i\mid {i\in [m]}\}.$$
  The connectivity graph is defined as $G=(V,E)$ with
  $$E=\bigcup_{i\in [n]}\{\{x_i, T\}, \{x_i, F\},\{\bar{x}_i, T\}, \{\bar{x}_i, F\}, \{x_i, a_i\}, \{\bar{x}_i, a_i\}  \} \cup \bigcup_{i\in [m]}\{\{x , b_i\}\mid x \in c_i \}.$$
  We define the metric~$d$ as a graph metric for the graph $G' = (V, E')$ with $$E'=\bigcup_{i\in [n]}\{\{x_i, T\}, \{\bar{x}_i, T\}, \{x_i, F\}, \{\bar{x}_i, F\}, \{a_i, F\}\} \cup \bigcup_{i\in [m]}\{\{b_i, T\}\}.$$
  That is, the distance $d(x,y)$ of two points $x,y\in V$ is defined as the length of the shortest $x$-$y$-path in the unweighted graph~$G'$.
  Finally we specify the points $T$ and $F$ as the given centers in the assignment problem. Then the problem is to find two disjoint clusters with respect to the centers $T$ and $F$ that cover all points from $V$ with minimum radius. This construction is illustrated in Figure~\ref{fig:2centershardness}.

\begin{figure}
\centering
\begin{minipage}{.45\textwidth}
 \centering
 \includegraphics[width=.8\linewidth]{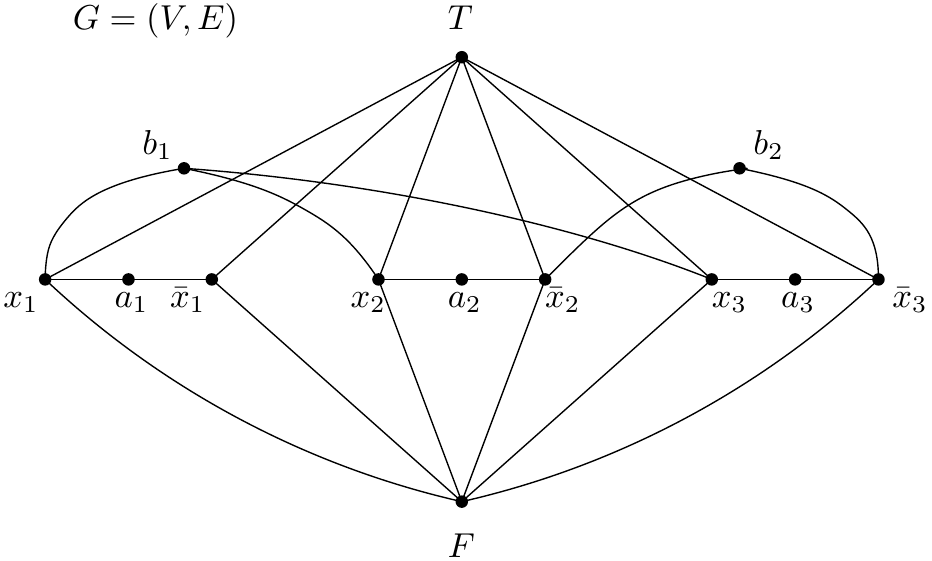}
\end{minipage}%
\hspace{0.45cm}
\begin{minipage}{.45\textwidth}
 \centering
 \includegraphics[width=.8\linewidth]{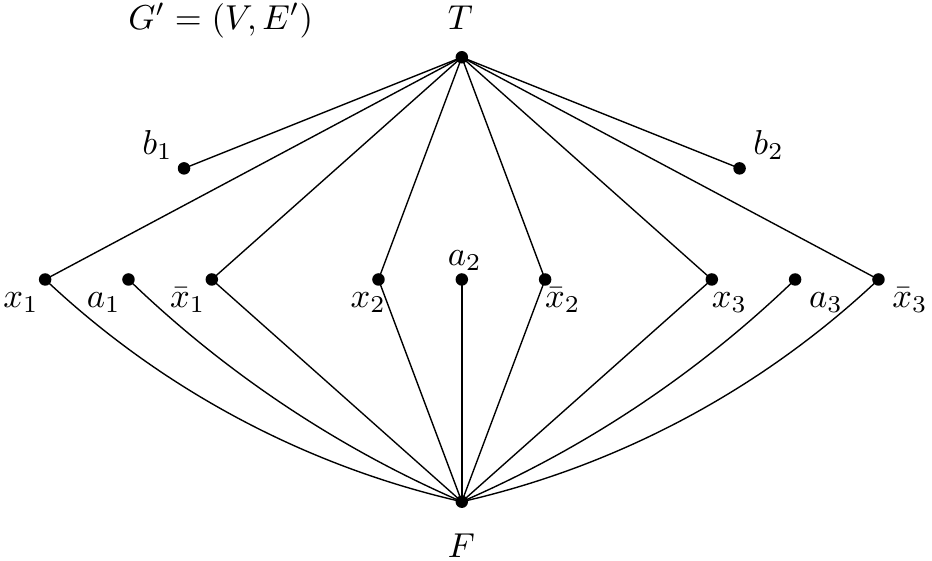}
\end{minipage}
\captionof{figure}{The connectivity graph~$G=(V,E)$ (left) and the graph metric~$G'=(V,E')$ (right) corresponding to the 3-SAT instance $(x_1 \lor x_2 \lor x_3) \land (\bar{x}_2 \lor \bar{x}_3)$.}
\label{fig:2centershardness}
\end{figure}

Let $x^*_1,\ldots, x^*_n$ be a satisfying assignment for the given 3-SAT instance. Then the following assignment provides a solution with radius $1$ for the connected $k$-center problem: if $x^*_i$ is true, then assign $x_i$ to center $T$ and assign $\bar{x}_i$ to center $F$, and otherwise, assign point $\bar{x}_i$ to center $T$ and assign point $x_i$ to center $F$; assign all points $a_i$ with $i\in [n]$ to center $F$ and all points $b_j$ with $j\in [m]$ to center $T$. It is easy to verify that the radius is $1$ according to the definition of $G'$ because all points connected to~$T$ or~$F$ are in $G'$ adjacent to $T$ or $F$, respectively. Since every point from~$V$ is assigned to either $T$ or $F$ the two clusters are disjoint and they cover the entire point set.

It remains to show that both clusters are connected. Since~$E$ contains the edges $\{x_i, T\}, \{x_i, F\},\{\bar{x}_i, T\}, \{\bar{x}_i, F\}$ for every $i\in [n]$, the points $x_i$ and $\bar{x}_i$ are directly connected to both point $T$ and $F$. Now consider a point~$b_j$ with $j\in [m]$. This point belongs to the cluster with center~$T$ but it is not directly connected to~$T$ by an edge. However, in the assignment $x^*_1,\ldots, x^*_n$ at least one literal of the corresponding clause~$C_j$ is satisfied. Let $x\in C_j$ be one such literal. Then according to the construction above $x$ is assigned to~$T$. Furthermore since $x\in C_j$, the connectivity graph contains the edge $\{x,b_j\}$, which implies that~$b_j$ is connected to $T$ in~$G$ via point~$x$. This shows that all points assigned to $T$ induce a connected subgraph of $G$. Now  consider point $a_i$ with $i\in [n]$. This point is assigned to~$F$. Since
$\{x_i, a_i\}, \{\bar{x}_i, a_i\} \in E$ and either $x_i$ or $\bar{x}_i$ is assigned to~$F$, we know that $a_i$ connects to a point $x\in \{x_i, \bar{x}_i\}$ which is assigned to $F$. This implies that all points assigned to $F$ induce a connected subgraph of $G$.

For the other direction, assume we have a feasible assignment solution of radius smaller than $3$. Observe that all points $a_i$ with $i\in [n]$ must be assigned to center $F$ and all points $b_j$ with $j\in [m]$ must be assigned to center $T$. Otherwise, the radius is at least $3$. Based on the assignment solution we construct a satisfying assignment $x^*_1,\ldots, x^*_n$ for the corresponding 3-SAT instance: if $x_i$ is assigned to $T$ then $x^*_i$ is set to true and to false otherwise. We claim that all clauses $c_i$ with $i\in [m]$ are satisfied by this assignment. 
Suppose that there exists a clause $c_j$ that is not satisfied. Since $\{b_j, T\} \notin E$ and point $b_j$ is assigned to center $T$, we may distinguish the following two cases:
 \begin{itemize} 
 \item There exists a neighbor point $x_i$ of $b_j$ in $G$ which is also assigned to $T$. By the definition of $E$, we know that $x_i$ is a literal of clause $c_j$. Because we set $x_i$ to be true, clause $c_j$ is satisfied.

 \item There exists a neighbor point $\bar{x}_i$ of $b_j$ in $G$ which is also assigned to $T$. Similarly by the definition of $E$, we know that $\bar{x}_i $ is a literal of clause $c_j$. Since point $a_i$ is assigned to center $F$ and $a_i$ is only neighbored to points $x_i$ and $\bar{x}_i$ in $G$, it follows that $x_i$ must be assigned to $F$. Thus $x_i$ is set to be false which in turn means that clause $c_j$ is satisfied.  
 \end{itemize}
In both cases, clause $c_j$ is actually satisfied contradicting the assumption that it is not satisfied. Hence the constructed solution satisfies the entire 3-SAT instance. 
 
By combing both directions we have proven that there exist a solution with radius $1$ for the assignment problem for the connected $2$-center problem with disjoint clusters if and only if the 3-SAT formula can be satisfied. Furthermore, with respect to the given centers~$T$ and~$F$ any clustering has either radius~$1$ or~$3$ because the points $x_i$ and $\bar{x}_i$ for $i\in[n]$ have a distance of~$1$ from both~$T$ and~$F$ while $a_i$ has distance~$1$ from~$F$ and distance~$3$ from $T$ and $b_j$ for $j\in[m]$ has distance~$1$ from~$T$ and distance~$3$ from~$F$. Altogether this implies that any approximation algorithm with an approximation factor smaller than~$3$ computes an assignment with radius~$1$ if the corresponding 3-SAT formula is satisfiable. We have seen that any such assignment can be transformed into a satisfying assignment of the 3-SAT formula. This proves the theorem.
\end{proof}

Observe that the previous theorem shows that the approximation algorithm for the assignment problem for connected $2$-center clustering with disjoint clusters from Corollary~\ref{cor:3and4approx-2givencenters} is optimal.

\subsection{Hardness of Connected \texorpdfstring{$k$}{k}-center with Freely Chosen Centers}
\label{lower:bound:three}

Now we consider the connected $k$-center problem with disjoint clusters where the centers are not given in advance. Our first observation is that the reduction in the proof of Theorem~\ref{lem:worst_assign} yields also for this case a hardness result for $k=2$. In that proof, the centers $T$ and $F$ are fixed. Assume that this is not the case and the algorithm can choose arbitrary centers. If the algorithm does not choose $T$ as a center then the distance of the points $b_j$ to the closest center is at least~$2$ and if the algorithm does not choose $F$ as a center then the distance of the points $a_i$ to the closest center is at least~$2$. Hence, if the 3-SAT formula is satisfiable then any approximation algorithm with an approximation factor smaller than~2 will compute a clustering with centers~$T$ and~$F$ and radius~1. This clustering corresponds to a satisfying assignment of the 3-SAT formula. This leads the following corollary, which has already been proven by Ge et al.~\cite{GeEGHBB08}.

\begin{corollary}
  \label{thm:lower_bound_2-freecenter}
 It is NP-hard to approximate the connected $k$-center problem with disjoint clusters with an approximation factor smaller than $2$ even if $k = 2$.
 \end{corollary}

Observe that the previous corollary shows that the approximation algorithm for the connected $2$-center problem with disjoint clusters from Corollary~\ref{cor:2approx-2freecenters} is optimal. Next we show a better lower bound for $k\ge 4$.

\begin{theorem}
 \label{thm:lower_bound_4-center}
It is NP-hard to approximate the connected $k$-center problem with disjoint clusters with an approximation factor smaller $3$ even if $k = 4$.
\end{theorem}
\begin{proof}
  Again we use 3-SAT as a starting point and reduce from 3-SAT to the connected $k$-center problem with $k=4$. Assume that a 3-SAT instance with variables $X = \{x_1,...,x_n\}$ and clauses $C=\{c_1, c_2, ..., c_m\}$ is given, where each clause consists of at most three literals from $\{x_i, \bar{x}_i\}_{i\in [n]}$. 
  We create an instance of the connected $k$-center problem that consists of $5$ sub-instances where each sub-instance $i\in [5]$ is a copy of the instance constructed in the proof of Theorem~\ref{lem:worst_assign}. We denote the points $T$ and $F$ from sub-instance~$i$ by $T_i$ and~$F_i$ and we link the sub-instances by these points. Figure~\ref{fig:4centershardness} shows this construction.

  Formally, we define the connectivity graph $G=(V, E)$ with $V=\bigcup_{i\in [5]} V_i$ where  $$V_i=\{T_i, F_i\}\cup \{x_{ij}, \bar{x}_{ij}, a_{ij}\mid {j\in [n]}\}\cup \{b_{ij}\mid {j\in [m]}\}$$
  $$T_1=F_4=T_5, F_1=T_2, F_2=T_3=F_5, F_3=T_4,$$
  and $E=\bigcup_{i\in [5]} E_i$ where each $E_i$ denotes the edge set from the instance constructed in Theorem~\ref{lem:worst_assign}. That is
  $$E_i=\bigcup_{j\in [n]}\{\{x_{ij}, T_i\}, \{x_{ij}, F_i\},\{\bar{x}_{ij}, T_i\}, \{\bar{x}_{ij}, F_i\}, \{x_{ij}, a_{ij}\}, \{\bar{x}_{ij}, a_{ij}\}  \} \cup B_i$$ where $B_i$ contains an edge $\{x_{ij},b_{i\ell}\}$ when $x_j\in c_{\ell}$ and an edge $\{\bar{x}_{ij},b_{i\ell}\}$ when $\bar{x}_j\in c_{\ell}$.
  The metric $d$ is again a graph metric defined by the graph $G' = (V, E')$ with $E'=\bigcup_{i\in [5]} E'_i$  where $$E'_i=\bigcup_{j\in [n]}\{\{x_{ij}, T_i\}, \{\bar{x}_{ij}, T_i\}, \{x_{ij}, F_i\}, \{\bar{x}_{ij}, F_i\}, \{a_{ij}, F_i\}\} \cup \bigcup_{j\in [m]}\{\{b_{ij}, T_i\}\}.$$
  That is, the distance $d(x,y)$ of two points $x,y\in V$ is defined as the length of the shortest $x$-$y$-path in $G'$. 
 The problem is to find four disjoint clusters that cover all points from $V$ with minimum radius. 

   \begin{figure}
     \centering
     \begin{minipage}{.45\textwidth}
       \centering
       \includegraphics[width=.8\linewidth]{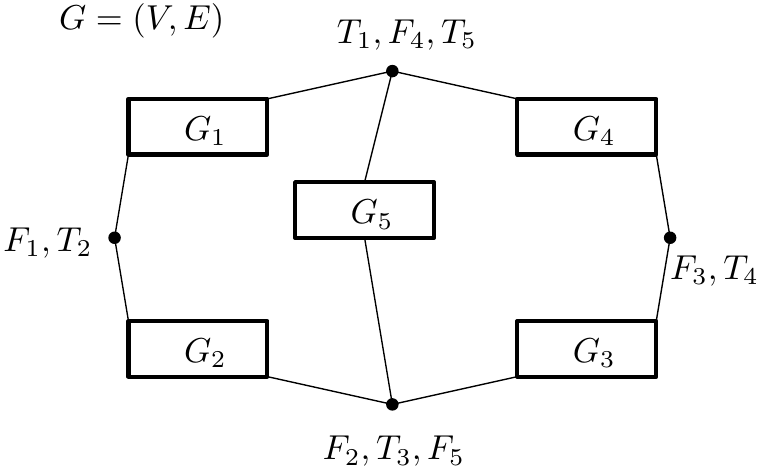}
     \end{minipage}%
     \hspace{0.45cm}
     \begin{minipage}{.45\textwidth}
       \centering
       \includegraphics[width=.8\linewidth]{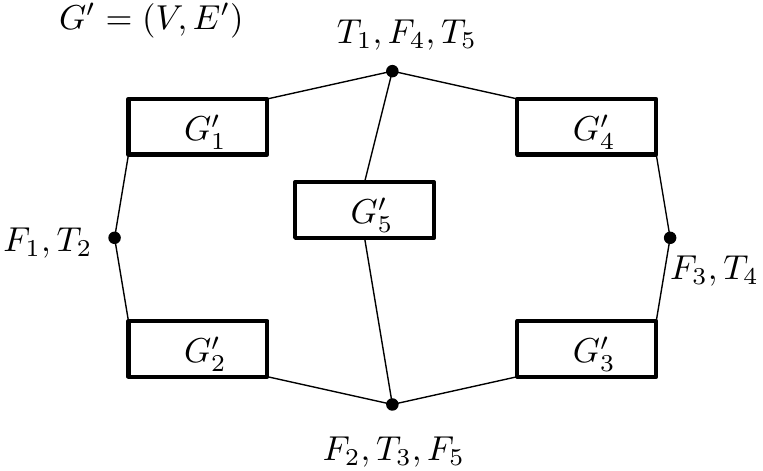}
     \end{minipage}
     \captionof{figure}{The connected $k$-center instance corresponding to a 3-SAT instance (left: connectivity graph~$G$; right: graph~$G'$ that induces the graph metric~$d$). For~$i\in[5]$, the connectivity graph $G_i$ and the graph $G'_i$ are copies of the corresponding graphs constructed in Theorem~\ref{lem:worst_assign}.}
     \label{fig:4centershardness}
     \end{figure} 

 If $x^*_1,\ldots,x^*_n$ is a satisfying assignment for the 3-SAT instance then the following is a solution with radius $1$ for the connected $4$-center problem with disjoint clusters: Choose all points $T_i$ and $F_i$ as centers (observe that due to the overlap these are exactly four centers). If $x^*_j$ is true then assign the points as follow: 
 \begin{itemize}
   \item
    Assign $x_{1j}$,  $\bar{x}_{4j}$, and $x_{5j}$ to center $T_1=F_4=T_5$. 
   \item 
   Assign $\bar{x}_{1j}$ and $x_{2j}$ to center $F_1=T_2$.
   \item 
   Assign $\bar{x}_{2j}$, $x_{3j}$ and $\bar{x}_{5j}$ to center $F_2=T_3=F_5$. 
   \item 
  Assign $x_{4j}$ and $\bar{x}_{3j}$ to center $F_3=T_4$.
 \end{itemize}
Otherwise if $x^*_j$ is false then assign the points as follow: 
 \begin{itemize}
   \item 
   Assign $\bar{x}_{1j}$,  $x_{4j}$, and $\bar{x}_{5j}$ to center $T_1=F_4=T_5$. 
   \item 
   Assign ${x}_{1j}$ and $\bar{x}_{2j}$ to center $F_1=T_2$.
   \item 
   Assign $ {x}_{2j}$, $\bar{x}_{3j}$ and $ {x}_{5j}$ to center $F_2=T_3=F_5$. 
   \item 
   Assign $\bar{x}_{4j}$ and $ {x}_{3j}$ to center $F_3=T_4$.
 \end{itemize}
 For the remaining points, apply the following assignment:
 \begin{itemize}
   \item Assign points $ \bigcup_{j\in [m]} \{b_{1j},a_{4j},b_{5j}\}$ to center $T_1=F_4=T_5$. 
   \item Assign points $ \bigcup_{j\in [m]} \{a_{1j},b_{2j}\}$ to center $F_1=T_2$.
   \item Assign points $ \bigcup_{j\in [m]} \{a_{2j},b_{3j},a_{5j}\}$ to center $F_2=T_3=F_5$. 
   \item Assign points $ \bigcup_{j\in [m]} \{b_{4j},a_{3j}\}$ to center $F_3=T_4$.
 \end{itemize}
 As in the proof of Theorem~\ref{lem:worst_assign}, one can easily verify that the radius of this clustering is~$1$ because each point from $V$ is connected by an edge from $E'$ to its corresponding center. Also similar to the reasoning in Theorem~\ref{lem:worst_assign}, we can show that the points that are assigned to the same center induce a connected subgraph of $G$, which is pairwise disjoint from the other subgraphs. Thus, we have a feasible 4-clustering with radius $1$ with disjoint clusters. 

 For the other direction, assume we have a feasible clustering of radius smaller than $3$. Since there are at most four centers, we know that there exists an index $h\in [5]$ such that all points in $V_h\setminus \{T_h, F_h\}$ are not centers.  
 According to the graph $G'$, for any point $v$ in $V \setminus (V_h \setminus\{T_h,F_h\})$, one of the following two statements holds:
 \begin{itemize}
 \item For all $i \in [n]$, the distance  $d(a_{hi},v) \geq 3$.
 \item For all $i \in [m]$, the distance   $d(b_{hi},v) \geq 3$. 
 \end{itemize}

 Suppose $T_h$ is assigned to center $v_t$ (it could be $T_h$ itself) and $F_h$ is assigned to center $v_f$ (it could be $F_h$ itself). Note that every point in $V_h$ needs to get assigned to one of these two centers because in the connectivity graph $G$ 
the points from $V_h \setminus \{T_h, F_h\}$ are connected to the other nodes only via $T_h$ and $F_h$.
 By the observation above, the center $v_t$ is at  distance at least $3$ either from each point in $A_h := \bigcup_{j\in [n]}\{a_{hj}\}$ or from each point in $B_h := \bigcup_{j \in [m]}\{b_{hj}\}$. The same is true for the center $v_f$. 
 If both centers are at distance at least $3$ from the same set of points $A_h$ or $B_h$, then the radius of the cluster containing points $A_h$ or $B_h$ is obviously at least $3$. Thus, we may assume without loss of generality that $v_t$ is at  distance at least $3$ from every node in $A_h$ and $v_f$ is at distance at least $3$ from every node in $B_h$. (Note that $T_h$ and $F_h$ have the same neighbors in the connectivity graph $G_h=(V_h, E_h)$.) Any clustering with radius smaller than $3$ needs to assign all nodes $A_h$ to $v_f$ and assign all nodes $B_h$ to $v_t$.

We know that $T_h$ is assigned to point $v_t$ and $F_h$ is assigned to point $v_f$. Furthermore all points $A_h$ (and $B_h$) are connected to $v_t$ (respectively $v_f$) via points $\bigcup_{j\in [n]}\{x_{hj}, \bar{x}_{hj}\}$ and $T_h$ (respectively $F_h$). Hence, the problem becomes an assignment problem for connected $k$-center with disjoint clusters and two fixed centers $T_h, F_h$. This is exactly the same problem that we analyzed in the proof of Theorem~\ref{lem:worst_assign}. There we argued already that a feasible assignment solution with radius smaller than $3$ leads to a satisfying assignment of the formula $C=\{c_1,\ldots, c_m\}$. 
 
Altogether this implies that any approximation algorithm with an approximation factor smaller than~$3$ computes an assignment with radius~$1$ if the corresponding 3-SAT formula is satisfiable. We have seen that any such assignment can be transformed into a satisfying assignment of the 3-SAT formula. This proves the theorem.
\end{proof}

\subsection{Hardness Results for Non-disjoint Connected Clustering}
\label{sec:non-disjoint}
In this section, we provide the proofs that lead to Corollary~\ref{thm:nondisjoint:both}.

\thmNonDisjointBoth*

The corollary follows immediately from Lemma~\ref{thm:tree_lowerbound_diameter_non-disjoint} and Lemma~\ref{thm:lower_bound_non-disjoint_k-center}.

\begin{lemma}
    \label{thm:tree_lowerbound_diameter_non-disjoint}
    Let $\epsilon>0$.
    Assuming P~$\neq$~NP, there is no $(2-\epsilon)$-approximation algorithm for the connected k-diameter problem with non-disjoint clusters, even if $G$ is a star (a tree of height 1).
\end{lemma}
    \begin{proof}
    We reduce the clique cover problem to the connected $k$-diameter problem with non-disjoint clusters on a star. Given a graph $G=(V,E)$ and a number $k<n:=|V|$, the clique cover problem is to decide if it is possible to find a partition of the vertices of the graph into at most $k$ cliques, where a clique is a subset of vertices within which every two vertices are adjacent. This problem is NP-complete~\cite{karp1972reducibility}.

    For a given instance of the clique cover problem, create an instance of the connected $k$-diameter problem on a star.
    By \emph{star} we refer to a tree of height $1$ with a root $r$ and leaves $v_1,\ldots,v_n$.
    We use the given nodes in $G$ as $v_1,\ldots,v_n$ and define the metric on these nodes by $d(r,v_i)=1$ for all $i\in[n]$, $d(v_i, v_j)=1$ for all $i, j\in [n]$ for which $\{v_i, v_j\}$ is an edge in graph $G$ and $d(v_i, v_j)=2$ for all other $i, j\in [n]$.

    If the instance for clique cover is a yes-instance, then there exist $k$ cliques $C_1, \ldots, C_k$ that together cover all nodes in $V$. If we add $r$ to every clique, then we get $k$ clusters in the star with diameter $1$.
    Now assume that there is a $k$-clustering of the star with diameter $1$. If we remove $r$ from all clusters in which it is contained, we obtain $k$ cliques that cover~$V$.

    Thus, solving the clique cover problem is equivalent to deciding whether the transformed instance allows for a $k$-clustering of diameter $1$. Furthermore, the only possible diameters are $1$ and $2$. Thus it is NP-hard to approximate the connected $k$-diameter problem with non-disjoint clusters on star connectivity graphs better than by a factor of $2$.
\end{proof}

\begin{lemma}
   \label{thm:lower_bound_non-disjoint_k-center}
   Let $\epsilon>0$.
   Assuming P~$\neq$~NP, there is no factor $(2-\epsilon)$-approximation algorithm for the connected k-center problem with non-disjoint clusters, even if $G$ is a star (a tree of height 1).
\end{lemma}
\begin{proof}
  We reduce the unweighted set cover problem to the connected $k$-center problem.
 Given a set of elements $E=\{e_1,\ldots,e_n\}$, a set of subsets $S_1,\ldots,S_m \subseteq E$ and a number $k\in\mathbb{N}$, the unweighted set cover problem asks whether there exists a set $I\subseteq\{1,\ldots,m\}$ of cardinality $k$ such that $\cup_{j \in I} S_j = E$. We assume that for any $e_i\in E$, there is at least one set that contains it. We define the connectivity graph $G=(V,E)$ by $V=\{z,v_1,\ldots,v_n,w_1,\ldots,w_m\}$ and $E=\{ \{ z,v_i\} \mid i = 1,\ldots n\} \cup \{ \{ z,w_i\} \mid i = 1,\ldots m\}$. The $v_i$ are supposed to represent the elements and the $w_i$ represent the sets. Notice that there are only edges to $z$ which means that any cluster satisfying the connectivity requirement can either only contain a single node or it has to contain $z$. Since we are in the non-disjoint case, it is allowed to put $z$ in multiple clusters.

We define a distance function which assigns either $1$ or $2$ to every point pair, which automatically constitutes a metric. We set $d(v_i,z)=2$ for all $i \in [n]$ and $d(w_i,z)=1$ for all $i\in[m]$. Furthermore, we set $d(v_i,w_j)=1$ if and only if $e_i \in S_j$. Finally, $d(v_i,v_j)=2$ for all $i,j\in[n]$ and $d(w_i,w_j)=1$ for all $i,j\in[m]$.  We keep $k$. This describes our connected $k$-center instance.

Let $I$ be a valid set cover. We choose $C=\{w_i \mid i \in I\}$ as our centers. Notice that $|C|=k$. We place $z$ in all $k$ clusters, it has a distance of $1$ to all centers. With respect to the connectivity constraint, we can now place the remaining points arbitrarily because they will be connected to the center via $z$. We can construct a clustering of cost $1$ because every vertex has a center within distance $1$: The remaining $w_j$ that are not centers can be placed arbitrarily because they are at distance $1$ of all centers. For every $v_i$, there is a set $S_j$ with $j \in I$ which covers $e_i$, which means that $w_j$ is a center and $d(v_i,w_j)=1$, so we can place $v_i$ in the cluster with center $w_j$ (ties are broken arbitrarily). Thus, we get a solution of radius $1$.

For the other direction, assume we have a solution of radius $1$. If there is a center at $z$, then the respective cluster can only contain points from $\{w_1,\ldots,w_m\}$ because the other nodes have distance $2$ to $z$. We can ignore such a cluster. If there is a center at a node $v_i$, then the cluster cannot contain $z$, again because $d(z,v_i)=2$ for all $i\in [n]$. A cluster not containing $z$ can only contain one point. Thus, all clusters of this type constitute singleton clusters. We can now construct our set cover. For any $w_j$ in the center set, we pick $S_j$. For any $v_i$ in the center set, we pick an arbitrary set that covers $e_i$. Together, these are at most $k$ sets that cover $E$.

We conclude that there is a set cover of size $k$ if and only if there is a connected $k$-center clustering of radius $1$. Furthermore, the only possible radii are $1$ and $2$, implying that it is NP-hard to approximate the connected $k$-center problem with star connectivity graphs better than by a factor of $2$.
\end{proof}

\subsection{Hardness for Connected \texorpdfstring{$k$}{k}-diameter on Star Connectivity Graphs}\label{sec:hardness:star:diameter}
\newcommand{\xy}{UMMS}

\lemLowerBound*
\begin{proof}
By \emph{star} we refer to a tree of height 1 with a root $r$ and leaves $v_1,\ldots,v_n$.
We reduce the uniform minimum multicut problem on stars (\xy) to the connected $k$-diameter problem.
Given a star, a list of $m$ vertex pairs $\{v_i,v_j\}$ and a number $k < n$, the \xy\ problem is to decide if it is possible to separate all pairs of vertices in the list by deleting at most $k$ edges. This problem is NP-hard (Theorem 3.1 in~\cite{GVY97}; we can assume without loss of generality that $r$ is not in any pair).

For a given instance of \xy, create an instance of connected $k$-diameter clustering by using the given star as $G$ and defining the metric on $V$ by $d(r,v_i)=1$ for all $i\in[n]$, $d(v_i, v_j)=2$ for all $i,j\in [n]$ for which $\{v_i,v_j\}$ is a pair and $d(v_i,v_j)=1$ for all other $i,j \in [n]$.

If the original instance is a yes-instance, then there exist $k$ edges $\{r,v_{i_1}\},\ldots,\{r,v_{i_k}\}$ such that all pairs are disconnected. This means that if we place every $v_{i_j}$ into its own cluster, then we get $k+1$ clusters and the maximum diameter of any cluster is $1$.
Now assume that there is a solution with $k+1$ clusters and a maximum diameter of $1$. Only one cluster can contain $r$ and the other $k$ clusters have to be singleton clusters because of the connectivity constraint. The edges to these singletons are sufficient to disconnect all pairs.
Thus, solving \xy\ is equivalent to deciding whether the transformed instance allows for a $(k+1)$-clustering of diameter $1$. Furthermore, the only possible diameters are $1$ and $2$. Thus it is NP-hard to approximate the connected $k$-diameter problem with star connectivity graphs better than by a factor of $2$.
\end{proof}

\section{Connected \texorpdfstring{$k$}{k}-diameter with Line Connectivity Graphs}\label{sec:lineproofs}

In this section we show Corollary~\ref{thm:line:all}.

\thmLineAll*
\begin{proof}
We proved the non-disjoint case for $k$-center in Section~\ref{sec:resultssummary}. Notice that that proof did not need the triangle inequality and holds for arbitrary distances. The disjoint case for $k$-center follows from Theorem~\ref{thm:opt_kcenter_tree}. For disjoint and non-disjoint $k$-diameter clustering, consider 
 Lemma~\ref{thm:opt_diam_line}.
\end{proof}

\begin{lemma}
    \label{thm:opt_diam_line}
    When the connectivity graph $G$ is a line graph, then the connected $k$-diameter problem with disjoint or non-disjoint clusters can be solved optimally in time $O(n^2\log n)$. This is true even if the distances are not a metric.
\end{lemma}
\begin{proof}
 We are given an instance where the connectivity graph is a line $L=(V, E)$ with vertices $V=\{v_1, v_2, ..., v_n\}$ and edges $E=\{\{v_i, v_{i+1}\}\mid i\in [n-1]\}$. Assume that also a diameter $r$ is given.
We first notice that every connected subset of $V$ is also a path. Then we observe that -- also for the non-disjoint case -- the line graph always allows for an optimal solution in which all clusters are disjoint: Assume there are two clusters that are not disjoint. Since the clusters are paths, this means that either one of the paths is a sub-path of the other, or the ends of the clusters overlap. In the first case, remove the cluster that is smaller from the solution. In the second case, delete the overlapping sub-path from one of the clusters. Since it is at the end of the path, the cluster stays connected, and its diameter will only get smaller.
(Remark: This does not work for $k$-center because a center may be in the overlapping region.)

For a vertex $v_i$, let $\Path(v_i, r)=\{v_i, ..., v_h\}$ be the longest path starting at $v_i$ which satisfies that $d(v_{\ell_1},v_{\ell_2}) \le r$ for all $\ell_1,\ell_2\in\{i,\ldots,h\}$. We say that $h-i+1$ is the \emph{length} of that path.
We start at $v_1$ and compute $\Path(v_1,r)$. It ends at $v_h$, so to find the next cluster, we compute $\Path(v_{h+1},r)$. We iterate this procedure until we reach $v_n$.
The procedure finds the minimum number of subpaths into which $L$ can be cut under the condition that the diameter does not exceed $r$, so if there is a solution with $k$ clusters of diameter at most $r$, it will find it.
Notice that if the path computed at $v_i$ has length $\ell$, then computing it takes time $O(\ell^2)$. Since the lengths of all paths sum up to $n$, and since the square of a sum can only be larger than the sum of the squares of its summands, the total running time to solve Problem~\ref{p:subroutine} is $O(n^2)$.
Notice that we use no triangle inequality in the proof, so the statement is even true for non-metric connected $k$-diameter with a line connectivity graph.
\end{proof}

\section{Conclusions and Open Problems}
We studied the connected k-center and k-diameter problem and proved several new results on the approximability of different variants of these problems. In particular, we developed a general framework to obtain approximation algorithms for the disjoint versions of these problems that relies on the existence of well-separated partitions. While we obtain constant-factor approximations for $L_p$-metrics in constant dimension and metrics with constant doubling dimension, our general upper bound is $O(\log^2{k})$. Since all our lower bounds are constant, an obvious open question is to close the gaps between the upper and lower bounds. One possibility to approach this would be to derive better well-separated partitions. However, we also show that with our approach no bound better than $O(\log\log{k})$ can be shown.

\paragraph{Acknowledgements}

We thank J\"urgen Kusche and Christian Sohler for raising the problem and for fruitful discussion on the modeling. We also thank Xiangyu Guo for the discussion on the algorithm design and analysis.  
 
\bibliography{literature}
\bibliographystyle{plainurl}
\appendix
\section{A Counterexample to the Greedy Algorithm}
\label{sec:Counterexample}

Ge et al.\ \cite{GeEGHBB08} use CkC to denote the connected $k$-center problem with disjoint clusters and CkC' to denote the connected $k$-center problem with non-disjoint clusters. They present an algorithm (Algorithm 2 in Figure 4 in \cite{GeEGHBB08}) to transform a feasible solution for the CkC' problem into a feasible solution for the CkC problem. 

Let $V'_1, \ldots, V'_k$ be a solution for CkC' with corresponding centers $c_1, \ldots, c_k$. This solution is iteratively transformed into a feasible solution $V_1,\ldots,V_k$ of CkC with the same centers $c_1,\ldots,c_k$. In the first iteration, the nodes $V'_1$ connected to $c_1$ are assigned to the cluster $V_1$. Then, in each iteration $i\in\{2,\ldots, k\}$, the induced subgraph $G[V_i'\setminus\cup_{j=1}^{i-1}V_j]$ is considered and all nodes that are reachable from $c_i$ in this subgraph are added to $V_i$. Then the algorithm iterates over all nodes~$v$ from $(\cup_{j=1}^{i-1}V_j)\cap V_i'$ and adds all yet unassigned nodes connected to $v$ in $G[V_i']$ to the cluster $V_j$ that contains~$v$.

In Lemma~4.5 of \cite{GeEGHBB08} it is argued that this algorithm is guaranteed to find a feasible solution for the CkC problem with maximum radius at most $3r$ when given a feasible solution for the CkC' problem with radius~$r$. We believe that this claim is false and present a counterexample in the following. In our opinion the incorrect reasoning in the proof of Lemma~4.5 is the following sentence: \emph{Besides, we observe that the distance between $v$ and its corresponding center node $c_j$ is at most~$r$.} This sentence implicitly assumes that $v$ belongs to the cluster~$V_j'$ with center~$c_j$. However, $v$ could also be a node from a different cluster~$V_{\ell}'$ for some $j<\ell<i$ that was added to the cluster~$V_j$ when $V_{\ell}$ was computed. In that case the distance between $v$ and $c_j$ can be larger than~$r$. This can happen repeatedly, which leads to an approximation factor of $\Omega(k)$.

The counterexample is shown in Figure~\ref{fig:GeInstance}: The left graph depicts the metric (i.e., a graph metric where each edge has length~1). The graph in the right depicts the connectivity graph. Let $k=6$. It is easy to construct a feasible optimal clustering for CkC' with radius~1 with centers $u_1, u_2, u_3, u_4, u_5, u_6$. In this clustering the cluster with center~$u_i$ contains the nodes $v_{i-1}, v_i, v_{i+1}, s_{i1}, s_{i2}, s_{i3}$ (for $i=1$, $v_{i-1}$ is interpreted as $v_6$, and for $i=6$, $v_{i+1}$ is interpreted as $v_1$).  Figure~\ref{fig:onecluster} shows as an example the cluster with center $u_1$.

\begin{figure}[htp]
\centering
\includegraphics[width=.4\textwidth]{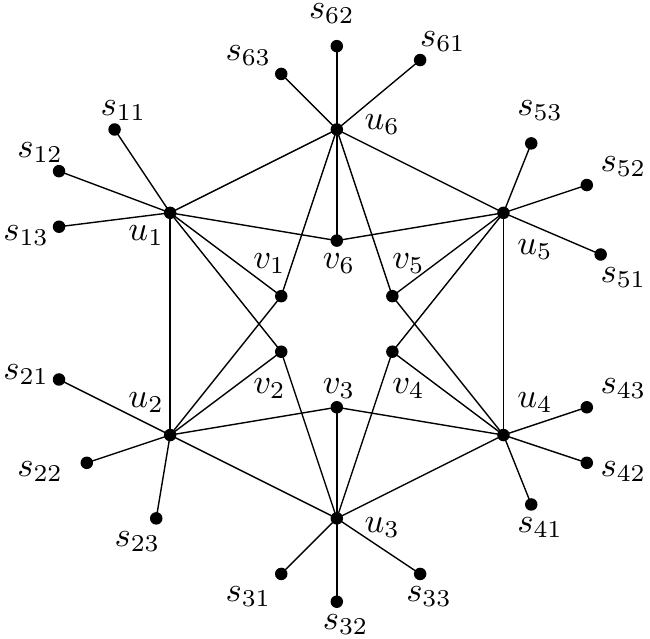}
  \hfill
\includegraphics[width=.4\textwidth]{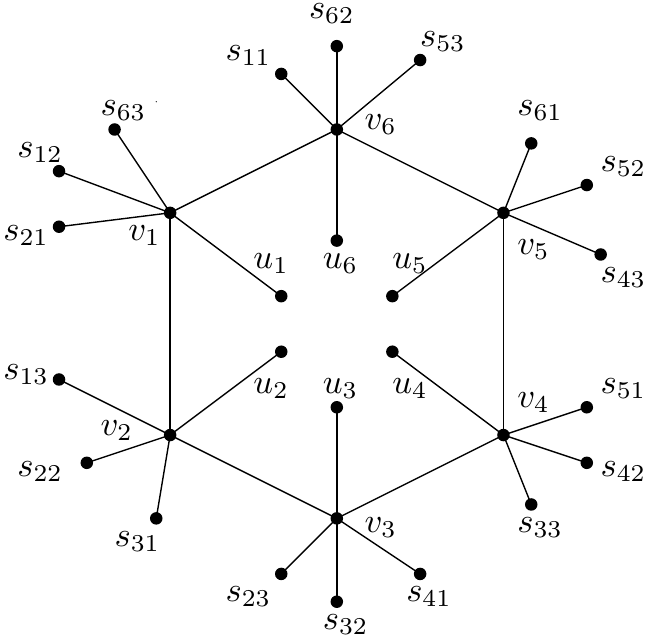}
\caption{Illustration of the metric (left) and the connectivity graph (right).
(Observe that the node positions of the nodes are different in the left and the right figure, i.e., the labels are not the same.)}
\label{fig:GeInstance}
\end{figure}

\begin{figure}[htp]
\centering
\includegraphics[width=.4\textwidth]{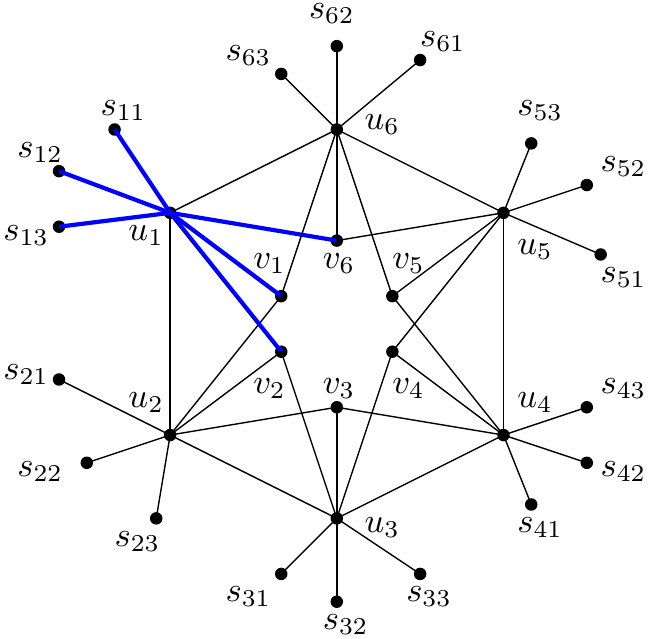} 
 \hfill
\includegraphics[width=.4\textwidth]{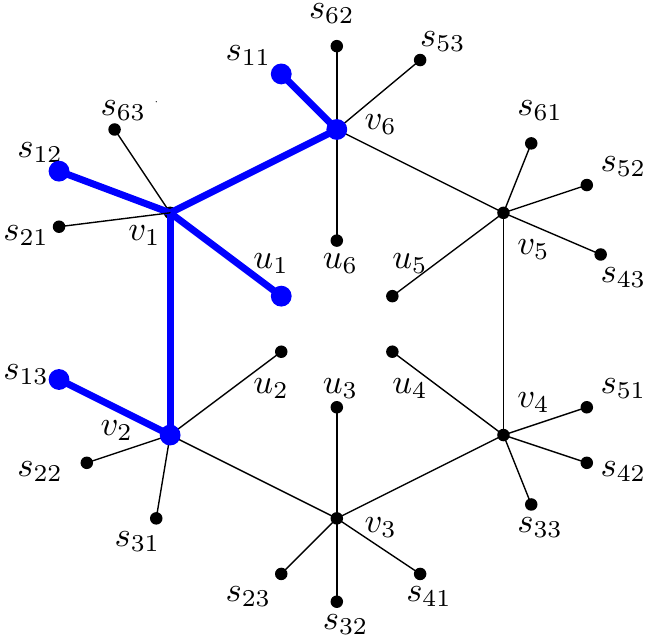} 
\caption{Illustration of one cluster in the metric (left), and the cluster in the connectivity graph (right).} 
\label{fig:onecluster}
\end{figure}

Now let's examine the solution computed by the algorithm due to Ge et al.:
\begin{itemize}
    \item For $i=1$, all nodes $\{v_1, v_2, v_6, s_{11}, s_{12}, s_{13}\}$ are assigned to $V_1$.
    \item For $i=2$, the center $u_2$ is isolated in $G[V_2'\setminus V_1]$. Hence $V_2$ consists only of $u_2$. The nodes $s_{21}, s_{22}, s_{23}, v_{3}$ from $V_2'$ get assigned to $V_1$ with center $u_1$ because these nodes are connected to $v_2$ which is already in~$V_1$.
    \item For $i=3$, the center $u_3$ is isolated in $G[V_3'\setminus (V_1\cup V_2)]$. Hence $V_3$ consists only of $u_3$. The nodes $s_{31}, s_{32}, s_{33}, v_{4}$ from $V_3'$ get assigned to $V_1$ with center $u_1$ because these nodes are connected to $v_3$ which is already in~$V_1$.
    \item \ldots
    \item At the end, all nodes except the centers are assigned to $u_1$ and the cluster radius is then $4$, which is more than $3$ times the radius of the solution for CkC'.  
\end{itemize}

The example shows that the disjoint clustering that is computed by the algorithm can be worse by more than a factor of~$3$ than the given non-disjoint clustering. By extending this construction appropriately, a bound of $\Omega(k)$ can be shown. Even more, the optimal disjoint clustering has radius~$2$ independent of $k$: Use each $u_i$ as a center and let the corresponding cluster be $\{u_i,v_i,s_{i2},s_{(i+1)1},s_{(i-1)3}\}$ (for $i=1$, the corresponding cluster is $\{u_1, v_1, s_{12}, s_{21},  s_{k3}\}$, and for $i=k$, the corresponding cluster is $\{u_k, v_k, s_{k2}, s_{11}, s_{(k-1)3}\}$). Hence, our construction also shows that the approximation factor of the algorithm is $\Omega(k)$ in the worst case.

\section{Connected Clustering with Tree Connectivity Graphs}
\label{chap:treeopt}

In this section, we solve Problem~\ref{p:subroutine} optimally for the connected $k$-center problem with disjoint clusters when the connectivity graph is a tree:
Given an unweighted tree $T=(V,E)$, a metric space $M=(V,d)$ with $d:V\times V\to \mathbb{R}$, a number $k\in \mathbb{N}$ and a radius $r$, find $k$ disjoint subtrees of $T$ (clusters) that cover all vertices and have maximum radius of $r$, if possible. We present an algorithm that uses  dynamic programming and finds a solution if a solution exists in polynomial time. This immediately implies an optimum algorithm for the connected $k$-center problem with disjoint clusters and a $2$-approximation algorithm for the connected $k$-diameter problem with disjoint clusters when the connectivity graph is a tree.

We obtained the results in this section before we were aware of the previous work on connected clustering. In fact, Ge et al.~\cite{GeEGHBB08} already proved that the $k$-center problem with disjoint clusters can be solved optimally by dynamic programming. Their algorithm is essentially the same as the one presented here. We keep our writeup here for completeness but would like to stress that the algorithm and the results in Sections~\ref{sec:TreeAlg} and~\ref{sec:TreeAlgCorrectness} are not new. In contrast, the assignment problem, which we discuss in Section~\ref{sec:TreeAlgAssignment}, has not been studied before.

\subsection{The Algorithm}\label{sec:TreeAlg}
We give an optimum algorithm with running time $O(n^2 \log n)$. The first observation is that since we know $r$, we can formulate the problem as a minimization problem that minimizes the number of sub-trees of radius $r$ that are necessary to cover $T$. Thus, \lq optimum\rq\ in the following refers to the smallest number of necessary centers (or sub-trees).
Before presenting the algorithm, we need some further notation.
\begin{itemize}
	\item We choose an arbitrary node as the root of tree $T$, let this be $o \in V$.
	\item For each node $a\in V$, let $\ch(a)$ be the set of all children of node $a$ in $T$.
	\item For each node $a\in V$, let $T_a=(V_a, E_a)$ denote the (maximum) subtree rooted at node $a$ not including $a$'s parent. Note that $T_o=T$.
	\item For two different nodes $u, v\in V$, let $d'(u, v)$ denote the maximum distance of $u$ to a node on the path connecting $u$ and $v$. 
\end{itemize}

The core of any dynamic programming is the definition of suitable sub-problems. We want to solve sub-problems for sub-trees and then combine them. However, if we just look at a vertex $a$ and the sub-tree $T_a$, then it is not clear how to decompose a solution into the part within the sub-tree and the part in the rest of the tree because vertices of $T_a$ may be assigned to a center in $T \backslash T_a$ and vice versa. We resolve this by considering pairs of vertices $a$ and $b$. $T_a$ is the sub-tree we are interested in, and $b$ is a center -- to which either $a$ or its parent is assigned to (details below). By this, we fix the interaction between $T_a$ and $T\backslash T_a$ in order to compute solutions for sub-trees independently and then combine them. Recall that our objective is to minimize the number of centers, so what we optimize over when solving a sub-problem is the number of sub-trees necessary to cover $T_a$. Now we define the sub-problems precisely.

\begin{itemize}
	\item Let $I(a, b)$ be the optimum value of the sub-problem defined on subtree $T_a$ under the condition that $b$ is a center and $a$ is assigned to it. The node $b$ has to be \emph{inside} of $T_a$. 
	This function will be computed for every node $a\in V$ and every center choice $b\in V_a$. There are $O(n^2)$ such relevant values. We further define $I(a)=\min_{b\in V_a} I(a, b)$ to be the optimum value of the subproblem defined on the subtree $T_a$ while root $a$ is assigned to any node inside of $T_a$. 
	\item Let $F(a, b)$ be the optimum value of the sub-problem defined on the subtree $T_a$ under the condition that $b$ is a center and the \emph{parent} of $a$ is assigned to $b$. The node $b$ has to be \emph{outside} of $T_a$.
	This function will be computed for every node $a\in V$ except the root $o$, and every center choice $b\in V\setminus V_a$.
	There are $O(n^2)$ such relevant values. 
\end{itemize} 

Why are these sub-problems sufficient? First notice that the connectivity constraint helps us with decomposing the tree. Assume a vertex $s$ in $T_a$ is connected to a center $b$ outside of $T_a$. Then
all nodes on the path connecting $s$ and $b$ have to be in the same cluster. In particular, $a$ and the parent of $a$ are in the same cluster as $s$. It thus suffices to look at all possible ways to assign the parent of $a$ to a center $b$ outside of $T_a$.

When no vertex in $T_a$ is connected to a center outside of $T_a$, then we need the reverse consideration: What information do we need about centers $T_a$ to later combine its solution with  $T\backslash T_a$? We observe that from the perspective of $T\backslash T_a$, only one center in $T_a$ is interesting: The center to which $a$ is connected. Any vertex $t$ from $T \backslash T_a$ that wants to connect to a center inside of $T_a$ needs to bed assigned to the center $a$ is connected to. Thus it suffices to look at all possibilities to assign $a$ to a center within $T_a$.

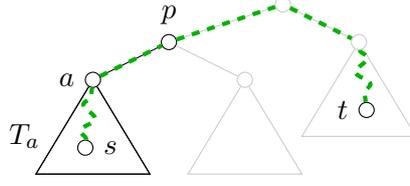
\begin{figure}
 \centering

\begin{tikzpicture}
 \node at (0,0) [circle, inner sep = 0cm, minimum width=0.2cm, draw, label=left:{$a$}] (a) {};
 \node at (1,0.5) [circle, inner sep = 0cm, minimum width=0.2cm, draw, label=above:{$p$}] (p) {};
 \node at (2.5,1) [circle, inner sep = 0cm, minimum width=0.2cm, draw=black!20] (gr) {};
 \draw (a) to (p);
 \draw [black!20] (p) to (gr);
 \draw (a) -- ($(a)+(-0.75,-1.25)$)  -- ($(a)+(0.75,-1.25)$) -- (a);
 \draw (a) -- (-0.75,-1.25) -- (0.75,-1.25) -- (a);
 \node at (-0.9,-0.75) {$T_a$};
 \node at (2,0) [circle, inner sep = 0cm, minimum width=0.2cm, draw=black!20] (gr2) {};
 \draw [black!20] (p) -- (gr2);
 \draw [black!20] (gr2) -- ($(gr2)+(-0.75,-1.25)$)  -- ($(gr2)+(0.75,-1.25)$) -- (gr2);
 \node at (3.5,0.5) [circle, inner sep = 0cm, minimum width=0.2cm, draw=black!20] (gr3) {};
 \draw [black!20] (gr3) -- ($(gr3)+(-0.75,-1.25)$)  -- ($(gr3)+(0.75,-1.25)$) -- (gr3);
 \draw [black!20] (gr) -- (gr3);

 \node at (-0.1,-0.9) [circle, inner sep = 0cm, minimum width=0.2cm, draw, label=right:{$s$}] (s) {};
 \draw [decorate, decoration=snake, dashed,green!70!black,ultra thick] (s) to (a);
 \draw [green!70!black,ultra thick, dashed] (a) to (p);
 \draw [green!70!black,ultra thick, dashed] (p) to (gr);
 \draw [green!70!black,ultra thick, dashed] (gr) to (gr3);
 \node at (3.6,-0.4) [circle, inner sep = 0cm, minimum width=0.2cm, draw, label=left:{$t$}] (t) {};
 \draw [decorate, decoration=snake, dashed,green!70!black,ultra thick] (gr3) to (t);
\end{tikzpicture}
\caption{An illustration of how the connectivity constraint makes the tree decomposable: If $s$ wants to connect to $t$, $a$ and $p$ also have to connect to $t$, and if $t$ wants to connect to $s$, then $p$ and $a$ also have to connect to $s$.}
\end{figure}

Now we provide the detailed recurrence relations between $F$ and $I$.
According to the definition of function $F(a, b)$ on nodes $a\in V, b\in V\backslash V_a$, we
compute it by the following case analysis.
\begin{enumerate}[label=(\roman*)]
	\item If $a$ is a leaf, only check if $a$ itself can be covered by $b$. If it can be assigned to $b$, no extra center is necessary to cover $T_a$. If not, we need $a$ to be a center. We get
	\begin{equation*}
		F(a,b) = \begin{cases}
					0  & \text{if } d'(b,a) \leq r \\[10pt]
					1  & \text{otherwise}.
				\end{cases}
	\end{equation*}

	\item If $a$ is not a leaf, then we want to recurse on the sub-trees defined by the vertices in $\ch(a)$. There are two cases depending on where $a$ is assigned to. If $a$ is assigned to $b$, then for all vertices in $\ch(a)$, the parent is assigned to $b$. Thus in this case, the minimum number of sub-trees to cover $T_a$ is the sum of $F(a',b)$ over all $a' \in \ch(a)$. If $a$ is not assigned to $b$, then no vertex in $T_a$ is assigned to a vertex outside of $T_a$. Thus the minimum number of sub-trees to cover $T_a$ is $I(a)$. It may be that option one is not available if $a$ cannot be assigned to $b$. If both options are available, we take the minimum. We get
	\begin{equation*}
		F(a,b) = \begin{cases}
					\min \left\{\sum_{a'\in \ch(a)} F(a', b), I(a) \right\}  & \text{if } d'(b,a) \leq r \\[10pt]
					I(a)  &  \text{otherwise}.
				\end{cases}
	\end{equation*}
\end{enumerate}

We can compute $I(a, b)$ on nodes $a\in V, b\in V_a$ by the following case analysis.
\begin{enumerate}[label=(\roman*)]
	\item If $a=b$ then $a$ is chosen as a center and for all sub-trees defined by the vertices in $\ch(a)$, the parent is assigned outside (assigning $a$ to itself is outside of the sub-trees) and we have
	\[
	 I(a,a)= 1 + \sum_{a'\in \ch(a)} F(a', a).
	\]
	\item When $a\neq b$, then $a$ is assigned to a node $b$ from $V_a\setminus \{a\}$.
	Notice that not all $b$ are actually feasible because $a$ might not be able to connect to $b$ and then the definition of $I(a,b)$ makes no sense. We set $I(a,b)$ to $\infty$ in this case because $T_a$ cannot be covered by a finite number of centers under the condition that $a$ connects to a center which it cannot connect to. Otherwise, since $a$ is connected to $b$, we know for all sub-trees where the parent is connected to. We only need to consider one detail:
	Let $\ch_b(a)$ be the uniquely defined child of $a$ which lies on the path connecting $a$ and $b$ and distinguish between the sub-tree of $\ch_b(a)$ and the other sub-trees of vertices in $\ch(a)$. For $T_{\ch_b(a)}$, the parent of $\ch_b(a)$ is connected to a vertex in $T_{\ch_b(a)}$, so we need $I(\ch_b(a), b)$ centers to cover this sub-tree. For all other $a' \in \ch(a)$, the parent is connected to a vertex outside of the sub-tree, so we need $F(a',b)$ centers for these sub-trees. Overall, we get

	\begin{equation*}
		I(a,b) = \begin{cases}
			I(\ch_b(a), b) + \sum_{a'\in \ch(a)\setminus\{\ch_b(a)\}} F(a', b) &  \text{if }  d'(b,a) \leq r \\[10pt]
			\infty  & \text{otherwise}.
				\end{cases}
		\end{equation*} 	
\end{enumerate}

Due to the above recurrence relations, the function values $F(a, b)$, $I(a, b)$ and $I(a)$ can be computed easily if the function values for the children of $a$ are known. For this one only needs to compute the function values for all leaves (with height $1$) and then the nodes with larger height. Let $\height(a)$ be the height of vertex $a$ in a rooted tree, which is the length of the longest downward path to a leaf from vertex $a$. Using that $I(a)=1$ and $I(a, a)=1$ for all leaves $a$ in $T$, and $F(a, b)=0$ if $d'(b, a)\le r$, otherwise $F(a, b)=1$ for all $b\neq a$, we obtain the following algorithm to compute the smallest number of necessary centers.

\begin{algorithm}
	\caption{\textsc{DynamicProgramming}$(T,d',r)$}\label{alg:dp_tree}
	\KwIn{tree $T=(V, E)$, modified distances $d'(\cdot, \cdot)$, radius $r$}   
	Set $F(a, \cdot)$, $I(a)$ and $I(a, a)$ for all leaves $a$ in $T$\;

	\For{$i=2,...,\height(o)$}{
   		\ForAll{nodes $a$ with  $\height(a)=i$}{
   			\ForAll{$b \in T_a$}{
   				compute $I(a,b)$ by the recurrence relation;
   			}
   			compute $I(a)$ by taking the minimum over all $I(a,b)$;
   			
			\ForAll{$b \notin T_a$}{
				compute $F(a,b)$  by the recurrence relation;
			}
   		}
   }
	\KwOut{number of centers (or clusters) $I(r)$}
 \end{algorithm}

\subsection{Correctness and Running Time}\label{sec:TreeAlgCorrectness}

\begin{theorem}
	\label{thm:dp_opt_tree}
	Algorithm~\ref{alg:dp_tree} computes the smallest number of necessary centers. 	
\end{theorem}
\begin{proof}
    The correctness of the algorithm relies on the correctness of the recurrence relation. We show via induction over the height of the nodes the following invariant: for every node $a\in V$ and $b\neq a$, the function values
	$I(a)$, $I(a, b)$ and $F(a, b)$ obtained in Algorithm~\ref{alg:dp_tree} are consistent with the respective function definitions.

	\textbf{Base Case:} Consider the leaves, nodes with height $1$. The function values for a leaf $a$, $I(a)=1$ and $I(a, a)=1$ are trivially correct since one center is necessary and also sufficient to cover a singleton vertex set with a non-negative radius. According to the definition of $F(a, b)$ in which the parent of $a$ is assigned to center $b$ and the connectivity constraint, leaf node $a$ could either be assigned to $b$ or be chosen as a center. Thus the optimum value of the subproblem defined on the subtree $T_a$ and a node $b$ \emph{outside} of $T_a$ will be $0$ if $d'(b, a)\le r$, otherwise it will be $1$.

	\textbf{Inductive Case:} Assume that for every node with height at most $h$ the invariant holds. Now we show that for every node with height $h+1$ the invariant also holds.  Consider a node $a$ with $\height(a)=h+1$. 
   
	To calculate $I(a, b)$ for any $b\in V_a$, according to the definition, we know that $a$ is assigned to $b$ which is inside of subtree $T_a$. If $d'(b, a)>r$, the radius of the cluster containing center $b$ and node $a$ is larger than $r$ which results in an infeasible clustering and thus $I(a, b)$ is set to be $\infty$. When $d'(b, a)\le r$, all nodes (including the child $\ch_b(a)$) in the path between $a$ and $b$ will be assigned to $b$ as $a$ is assigned to $b$ because of the connectivity clustering property. Either we choose $a$ as a center, or we choose a node $b\neq a \in V_a$ to be the center of $a$. In the first case, we count center $a$ and a number of necessary clusters for each subtree rooted at a child of $a$ while $ a$ is assigned to $a$, thus $I(a, b)=I(a, a)=1+ \sum_{a'\in \ch(a)} F(a', a) $; and in the second case, we can describe $I(a, b)$ by summing up $I(\ch_b(a), b)$ and the number of necessary clusters for each subtree rooted at a child of $a$ except $\ch_b(a)$, i.e., $ \sum_{a'\in \ch(a)\setminus \{\ch_b(a)\}} F(a', b) $, since we know that the parent of child $a'\in \ch(a)\setminus \{\ch_b(a)\}$ is assigned to a node $b$ outside of subtree $T_{a'}$ and the parent of child $\ch_b(a)$ is assigned to a node $b$ inside of subtree $\ch_b(a)$.  Thus, $I(a, b)=I(\ch_b(a), b)+ \sum_{a'\in \ch(a)\setminus \{\ch_b(a)\}} F(a', b) $. Altogether this shows the correctness of the recurrence of $I(a, b)$, and it is calculated by the function values with respect to nodes $\ch(a)$ with lower height $h$. We also obtain the correct $I(a)$ immediately by definition.  

	Now we consider $F(a, b)$. By definition, the parent of $a$ is assigned to node $b$, a node outside of $T_a$. If $a$ itself cannot be reached from $b$ we need to assign $a$ to a node in $T_a$ and thus $F(a,b) = I(a)$. Otherwise we can either assign $a$ to a node in $T_a$, or assign $a$ to $b$ based on the connectivity clustering property. Thus, $F(a, b)=\min\{I(a), \sum_{a'\in \ch(a)} F(a', b) \}$ in which the first term represents the minimum number of clusters needed for tree $T_a$ while $a$ is assigned to a node in $T_a$, and the second term sums the minimum number of clusters needed for subtrees $T_{a'}$ as the parent of $a'$ is assigned to a node outside of $T_{a'}$. Thus, $F(a, b)$ is calculated by function values with respect to nodes $\ch(a)$ with lower height $h$. 
	
	The invariant also holds true for nodes with height $h+1$ as the invariant holds true for nodes with height $h$, establishing the induction step.
\end{proof}

	\begin{lemma}
		\label{thm:dp_time_tree}
		There exists an implementation of Algorithm~\ref{alg:dp_tree} with running time $O(n^2\log n)$ for the disjoint connected $k$-center problem on trees with $n$ nodes.
		\end{lemma}
    \begin{proof}

    We only need to argue that Algorithm~\ref{alg:dp_tree} needs $O(n^2)$ time to solve Problem~\ref{p:subroutine}.
    Precomputing $d'$ also requires a small dynamic programming approach to be achieved in time $O(n^2)$ by doing $n$ breadth-first-searches.

	The function $I(a,b)$ can be evaluted in time $O(1 + |\ch(a)|)$ because we might have to calculate a sum over all children of $a$ and all other operations need only constant time. The same holds for $F(a, b)$. Besides this it only remains to calculate $I(a)$ which is the minimum of $n$ values calculated before. Hence the entire running time to fill out the entries for a given node $a$ is upper bounded by $O(n (1 + |\ch(a)|)$.
	
	By summing this term up for every node in the tree we get that the total time complexity of our algorithm is upper bounded by
	\begin{equation*}
	O\left( \sum_{a \in V} n (1 + |\ch(a)|) \right) \subseteq O(n^2) + O\left(n \sum_{a \in V} |\ch(a)|\right) \subseteq O(n^2).
	\end{equation*}
	We used for the last step that every node has only one parent and thus is only contained in one of the $\ch(a)$.	
\end{proof}

We combine Theorem~\ref{thm:dp_opt_tree} and Theorem~\ref{thm:dp_time_tree} to obtain the following result.

\thmTreeOpt*

\subsection{An Optimal Assignment Algorithm for Trees}\label{sec:TreeAlgAssignment}

In this section, we provide an assignment algorithm for the connected $k$-center problem with disjoint clusters on trees. Let $C$ with $|C|=k$ be the set of given centers. The idea is to provide a subroutine that solves Problem~\ref{p:subroutine} optimally in $O(nk)$ time. In addition with checking only $O(\log{n})$ radii, this yields an optimal assignment algorithm with running time $O(nk \log n)$. 

We first simplify the $k$-center instance by assuming that all centers are leaves. This is without loss of generality as the following reasoning shows: Assume that there is a center $c\in C$ that is not a leaf and let $N\subseteq V$ denote the set of its neighbors (i.e., $N$ contains the parent and the children of $c$). Now we remove the node $c$ from the tree~$T$ and obtain by this operation $|N|$ disconnected trees. We add to each of these trees a copy of $c$ and connect it to the corresponding neighbor of~$c$. The copies of $c$ are leaves because they have degree one. We mark the copies of $c$ as centers and solve the assignment problem for each tree separately.

The idea of the assignment problem is relatively simple. Pick an arbitrary non-center point $t$ as root and process the points in a bottom-up way, from leaves to root. For each point $v\in V$ we save two sets: $n(v)$ and $z(v)$. The set $n(v)$ consists of all previously processed points that have to be served together with $v$, i.e., the points that should be assigned to a center that $v$ is assigned to. The set $z(v)$ consists of all reachable centers descending from $v$. Here, we say a center $c$ is reachable for $v$ if $c$ can be reached for all points in $n(v)$, i.e., $d(x, c) \leq r$ has to hold for all $x \in n(v)$ and for all $v\in n(z)$. The intuition of computing the above two sets for points strongly relies on the tree structure: since all centers are on the leaves and the two sets of each point only contain points descending from $v$, we can first easily compute them for all leaves, and then compute for all neighbors of these leaves and so on.  At the end, if $z(t)$ is empty, then no center is reachable for point $t$ and we concludes that $r$ is too small. Otherwise process the points from top to bottom and assign them to viable centers. 

A more precise description is given in Algorithm~\ref{alg:assignment-trees}. Let $\pi(v)$ be the parent point of $v$ in the tree and let $ P_{v, c}$ be the unique $v$-$c$-path in the tree. From line 1 to line 3, set the processed set $M=\emptyset$ and initialize all lists $n(\cdot), z(\cdot)$. From line 4 to line 9 and line 10 to 16, we update the two sets for all points bottom-up, from leaves to root. After processing a point $v$, if $z(v)$ is empty, then $v$ should be assigned to a center the same as its parent $\pi(v)$ and so we add all of $n(v)$ to $n(\pi(v))$. In lines 17-18 and line 19-24, the algorithm either outputs \emph{fail} implying that $r$ is too small, or an assignment solution. 
 
\begin{algorithm}
    \SetAlgoLined

    \KwIn{A tree $T = (V, E)$ with root $t$, centers $C$, a metric $d$, a radius $r$.}
    \KwResult{If there exists an assignment with radius $r$ then the algorithm will find such an assignment. Otherwise fail.}
    \BlankLine

    $M \gets \emptyset$\;
    $z(v) \gets \emptyset$ for all $v \in V$\;
    $n(v) \gets \{v\}$ for all $v \in V$\;

    \BlankLine

    \ForAll(){leaves $\ell \in V$}{
        $M \gets M \cup \{\ell\}$\;
        \eIf{$\ell \in C$}{
            $z(\ell) \gets z(\ell) \cup \{\ell\}$\;
        }{
            $n(\pi(\ell)) \gets n(\pi(\ell)) \cup n(\ell)$\;
        }
    }

    \BlankLine

    \While(){$M \neq V$}{
        pick $v \in V$ such that $u \in M$ for each child $u$ of $v$\;
        \ForAll(){child $u$ of $v$}{
            $z(v) \gets z(v) \cup \{c \in z(u) \; | \; \forall x \in n(v): d(x, c) \leq r\}$\;
        }
        \If(){$z(v) = \emptyset$}{
            $n(\pi(v)) \gets n(\pi(v)) \cup n(v)$\;
        }
        $M \gets M \cup \{v\}$\;
    }

    \BlankLine
    
    \If(){$z(t) = \emptyset$}{
        \textbf{fail}\;
    }

    \BlankLine

    \While{$M \neq \emptyset$}{
        let $v$ be the point in $M$ that was added last\;
        pick $c \in z(v)$ arbitrarily and let $P_{v, c}$ denote the unique $v$-$c$-path\;
        \ForAll(){$x \in P_{v, c}$}{
            assign all points in $n(x)$ to $c$\;
            $M \gets M \setminus n(x)$\;
        }
    }

    \Return{assignment}

    \caption{Assignment Algorithm}
    \label{alg:assignment-trees}
\end{algorithm}

\begin{lemma}
    \label{lem:assignment_tree_radius}
    With radius $r$, for any $v \in V$ and $c \in z(v)$ we have that $d(y, c) \leq r$ for all points $y \in \bigcup_{x \in P_{v, c}} n(x)$.
\end{lemma}

\begin{proof}
    For all $x \in P_{v,c}$, we know that $c \in z(x)$ since $c$ can only be passed upward along path $P_{v, c}$ (see lines 12 and 13) in a tree graph. If $z(x)$ does not contain $c$, then $z(v)$ can neither. In addition, line 13 ensures that the distance between each point in $n(x)$ and center $c$ should be at most $r$, which completes the proof.  
\end{proof}

Now we are ready to analyze the two possible outputs of the assignment algorithm with radius $r$: 
(1) The output is an assignment. According to the assignment process line 20-21 and line 22-24, for a (last added) point $v$ and the center $c$ chosen corresponding to $v$, all points assigned to the center $c$ form a subset of points in $\bigcup_{x \in P_{v, c}} n(x)$. Combining with the distance bound in Lemma~\ref{lem:assignment_tree_radius}, the radius of all points assigned to $c$ is at most $r$. 
(2) The output is fail. If $z(x)=\emptyset$ for some point $x\in V$, then there exists no center $c$ descending from $x$ that has distance at most $r$ to every point in $P_{c, x}$. In addition, since the subgraph induced by each output cluster has to be connected and the connectivity graph is a tree, point $x$ must be assigned to the same center as point $\pi(x)$. Since we process the points in a bottom-up fashion, from leaves to root, root $t$ is processed at the end, and thus $z(t)=\emptyset$ implies that there is no feasible assignment with radius $r$. We conclude that the assignment algorithm solves Problem~\ref{p:subroutine} optimally.

Note that Algorithm~\ref{alg:assignment-trees} needs at most $O(nk)$ time since we check the distance from every point to every center at most once. Immediately, we have the following theorem.

\begin{theorem}
    \label{thm:opt_assignment_trees_k-center}
    There exists an optimal assignment algorithm on trees for the connected $k$-center problem with disjount clusters that runs in $O(nk \log n)$ time, where $k$ is the number of centers.
\end{theorem}

\end{document}